\documentclass[10pt]{article}

\usepackage{makecell}

\usepackage{amssymb,amsmath,amsthm,sectsty,url}
\usepackage[letterpaper,hmargin=1.0in,vmargin=1.0in]{geometry}
\usepackage[pdftex,colorlinks,linkcolor=blue,citecolor=blue,filecolor=blue,urlcolor=blue]{hyperref}
\usepackage[capitalize,noabbrev]{cleveref}
\usepackage{color}
\usepackage[boxed]{algorithm}
\usepackage{amssymb,amsmath,amsthm,sectsty,url}
\usepackage{tikz}
\usepackage{graphicx}
\usepackage{nicefrac}
\usepackage{aliascnt}
\usepackage{colortbl}

\usepackage{subfig}
\usepackage{wrapfig}
\usepackage{algorithm}
\usepackage{algorithmic}
\usepackage{multirow}
\usepackage{multicol}
\usepackage{booktabs}

\usepackage[dvipsnames]{xcolor}
\usepackage{wrapfig}

\usepackage[para]{footmisc}

\usepackage{titling}

\pretitle{\begin{center}\Large\bfseries\rmfamily}
\posttitle{\par\vskip 1.5em\end{center}} 

\newenvironment{CompactEnumerate}{
\begin{list}{\arabic{enumi}.}{%
\usecounter{enumi}
\setlength{\leftmargin}{10pt}
\setlength{\itemindent}{1pt}
\setlength{\topsep}{1pt}
\setlength{\itemsep}{1pt}
}}
{\end{list}}

\newenvironment{CompactItemize}{
\begin{list}{\tiny$\bullet$}{%
\setlength{\leftmargin}{10pt}
\setlength{\itemindent}{10pt}
\setlength{\topsep}{5pt}
\setlength{\itemsep}{1pt}
}}
{\end{list}}

\newcommand{\eps}{\epsilon}

\newcommand{\addr}{\texttt{addr}}
\newcommand{\heapaddr}{\texttt{a}}
\newcommand{\bits}{b}

\newcommand{\mc}[1]{\mathcal{#1}}

\newcommand{\ouralg}{\textsc{CoDEQ}}
\DeclareMathOperator*{\argmax}{arg\,max}

\newtheorem{theorem}{Theorem}[section]
\crefname{theorem}{Theorem}{Theorems}

\newaliascnt{lemma}{theorem}

\aliascntresetthe{lemma}
\crefname{lemma}{Lemma}{Lemmas}

\newaliascnt{proposition}{theorem}
\newtheorem{proposition}[proposition]{Proposition}
\aliascntresetthe{proposition}
\crefname{proposition}{Proposition}{Propositions}

\newaliascnt{corollary}{theorem}

\aliascntresetthe{corollary}
\crefname{corollary}{Corollary}{Corollaries}

\newaliascnt{fact}{theorem}

\aliascntresetthe{fact}
\crefname{fact}{Fact}{Facts}

\newaliascnt{definition}{theorem}
\newtheorem{definition}[definition]{Definition}
\aliascntresetthe{definition}
\crefname{definition}{Definition}{Definitions}

\newaliascnt{remark}{theorem}

\aliascntresetthe{remark}
\crefname{remark}{Remark}{Remarks}

\newaliascnt{conjecture}{theorem}

\aliascntresetthe{conjecture}
\crefname{conjecture}{Conjecture}{Conjectures}

\newaliascnt{claim}{theorem}

\aliascntresetthe{claim}
\crefname{claim}{Claim}{Claims}

\newaliascnt{question}{theorem}

\aliascntresetthe{question}
\crefname{question}{Question}{Questions}

\newaliascnt{exercise}{theorem}

\aliascntresetthe{exercise}
\crefname{exercise}{Exercise}{Exercises}

\newaliascnt{example}{theorem}

\aliascntresetthe{example}
\crefname{example}{Example}{Examples}

\newaliascnt{notation}{theorem}

\aliascntresetthe{notation}
\crefname{notation}{Notation}{Notations}

\newaliascnt{problem}{theorem}

\aliascntresetthe{problem}
\crefname{problem}{Problem}{Problems}

\title{Quantization for Vector Search under Streaming Updates}

\author{%
 Ishaq Aden-Ali\thanks{UC Berkeley. Work done during an internship at Amazon}
  \and
  Hakan Ferhatosmanoglu\thanks{Amazon}
  \and
  Alexander Greaves-Tunnell\thanks{Amazon}
  \and
  Nina Mishra\thanks{Amazon}
  \and
  Tal Wagner\thanks{Amazon and Tel Aviv University}
}

\begin{document}
\maketitle

\begin{abstract}

Large-scale vector databases for approximate nearest neighbor (ANN) search typically store a quantized dataset in main memory for fast access, and full precision data on remote disk. State-of-the-art ANN quantization methods are highly data-dependent, rendering them unable to handle point insertions and deletions. This either leads to degraded search quality over time, or forces costly global rebuilds of the entire search index. In this paper, we formally study data-dependent quantization under streaming dataset updates. We formulate a computation model of limited remote disk access and define a dynamic consistency property that guarantees freshness under updates. We use it to obtain the following results: Theoretically, we prove that static data-dependent quantization can be made dynamic with bounded disk I/O per update while retaining formal accuracy guarantees for ANN search. Algorithmically, we develop a practical data-dependent quantization method which is provably dynamically consistent, adapting itself to the dataset as it evolves over time. Our experiments show that the method outperforms baselines in large-scale nearest neighbor search quantization under streaming updates.

\end{abstract}

\section{Introduction}

Approximate nearest neighbor (ANN) search is a fundamental task 
in modern retrieval contexts. 
Its significance has grown as vector search has increased in popularity with the rising quality of learned embeddings. 
In the context of large language models (LLMs), both retrieval-augmented generation (RAG) and search agents use ANN search to incorporate contextual and proprietary content.

ANN search systems often encode data elements into vectors and maintain a vector database for searching. 
Modern vector databases are {\em large} due to both the quantity of indexed documents and the dimensionality of embeddings.  These datasets tend to be too large or too expensive to fit in main memory.  Consequently, it is common to quantize data, sketching each point into a prescribed small number of bits. 
The quantized dataset is retained in main memory while full-precision vectors are stored on disk.  When a query arrives, the quantized data is used to identify results, either producing the final outputs or forming a candidate set that is then optionally retrieved from disk and re-ranked by exact distance to identify nearest neighbors. ANN algorithms must therefore carefully balance main memory vs. disk given the performance disparity and costs of communication between storage tiers.

Modern vector databases are often {\em dynamic} in the sense that new content continuously arrives while old content is removed. When there are differences between the new and old content, the result is a drifting distribution of both indexed data and queries. Baranchuk  et al.~\cite{dedrift} illustrate this phenomenon practically in the context of embedded image vectors for a photo database, where content changes over seasons.  
The quantization that works on a prefix of a stream (pictures in the summer) needs to change as the underlying data shifts over time (pictures in the fall). The topic was also the subject of the Streaming Track of the NeurIPS 2023 BigANN Competition~\cite{simhadri2024results}.

This raises the main challenge studied in the present work: \emph{how can we efficiently and reliably update a quantized vector database over a stream of insertions and deletions?} 

This question is related to the broader matter of data-dependence in ANN algorithms \cite{andoni2015optimal,andoni2015tight,quadsketch,jayaram2024data}. 
Quantization algorithms for ANN can be either \emph{data-oblivious} or \emph{data-dependent}. Loosely speaking, data-oblivious quantization quantizes each vector on its own, irrespective of the vector dataset it is part of. Data-dependent quantization quantizes each vector in the context of the dataset, and thus the same vector would be assigned different quantized representations within different datasets. We give a detailed overview of these notions in \Cref{app:datadependent}. 

Here, for concreteness, we mention  
Product Quantization (PQ)~\cite{jegou2010product, vqindex2002}, a widely-used data-dependent ANN quantization method based on k-means clustering. It partitions vector dimensions into subspaces, applies k-means clustering within each subspace, and represents vectors through the cross product of their nearest centroids in each subspace.
Consider trying to adapt PQ to the streaming ANN settings where vectors are inserted to and deleted from the vector database. 
This presents two primary challenges:  
\begin{CompactItemize}
\item[(i)] Updating quantization centroids to reflect evolving data distributions. 
\item[(ii)] Reassigning each point in the dataset to new centroid among the updated centroids. 
\end{CompactItemize}
Existing work addresses 
centroid updates for k-means clustering \cite{valiant,gmmmo}, but these methods can trigger widespread point reassignments across the dataset. 
For ANN quantization, they address challenge (i) but not challenge (ii), which might may render the computational cost per update linear in the entire dataset size. 
Such massive point-to-centroid reassignments are indeed unavoidable for $k$-means, where single point insertions or deletions can alter many assignments. Each reassignment requires disk access to retrieve the affected point (as its full precision is necessary to determine its nearest centroid)  and creates substantial I/O overhead that becomes prohibitive at scale. Therefore, vector quantization over a stream requires methods that not only can update cluster centroids, but also reassign only a \emph{limited} number of points per update to minimize disk access. 

In summary, at present, there is no theoretical framework for ANN quantization under streaming updates, and existing heuristic methods struggle either with scalability or with maintaining data freshness. The goal of the present paper is to address this gap at a foundational level.

\subsection{Our Contributions}
We undertake a principled study of data-dependent ANN quantization under streaming data updates.
To this end, we first formalize the \emph{Dynamic Sketches with Disk Access (DySk)} model: a computational model for hybrid memory-disk access with limited I/O, which captures ANN quantization.

We then use this model to obtain two results. Our first result is proving a worst-case quantization bound for $(1+\epsilon)$-ANN with dynamic updates, which essentially matches the best possible bound for \emph{static} quantization, while also supporting efficient updates with limited disk I/Os. 

Our second result is a practical dynamic quantization method, achieved by combining insights from the empirically successful PQ method with our theoretical framework and with new algorithmic ideas. 
As a building block for our method, we develop a new kind of memory-disk priority queues, with a focus on a \emph{constant} number of consecutive disk I/Os. 
This departs from prior work on memory-disk priority queues (e.g., in the external memory model), which has mostly allowed for a logarithmic number of consecutive I/Os. Since minimizing the number of consecutive I/Os is crucial in ANN search, we introduce these new data structures that support updates through polylog-size memory-disk data transfer with only $O(1)$ consecutive disk I/Os.

Empirically, we design an experimental framework for ANN search under streaming updates, and show that our method outperforms existing methods and maintains its level of recall as the stream of inserts and deletes evolves over time, while baselines decline.

\paragraph{Our techniques.}
Our worst-case quantization bound for dynamic $(1+\varepsilon)$-ANN is based on the quadtree-based static ANN quantization method of \cite{indyk2017near}. Their method is based on data-dependent lossy compression that discards certain portions of the tree identified as irrelevant to the points presently in the dataset. In the dynamic case, if a new point is inserted, then due to data-dependence, a discarded portion may become necessary again. We analyze the compression scheme under dynamic updates with remote disk access, and show that for every inserted point it is possible to identify a small portion of discarded data that can be recovered from the disk with a constant number of I/Os of polylogarithmic size. This maintains essentially the same static near-optimal quantization bound from \cite{indyk2017near} in the dynamic setting. The result is presented in \Cref{thm:quadsketch}.

Our practical method is based on kd-trees, a classical method for nearest neighbor search. 
However, our use kd-trees is different from their usual manner of use. Classically, their role is pruning the dataset at query time by limiting the search to a beam in the kd-tree around the query.
In contrast, we build a \emph{shallow} kd-tree -- its depth is the desired number of quantization bits per point -- and use the partition induced by the leaves as a clustering-based quantizer to reduce memory size. 

The reason a shallow kd-tree is useful for dynamic ANN quantization is that it maintains a data-dependent hierarchical partition of the dataset, which is induced by coordinate medians at each node. Crucially, a median-based partition is \emph{stable} under incremental updates: if a point is inserted or deleted, the coordinate-median of the dataset moves by at most one position, so only one point may need to be reassigned to the other side of the partition (i.e., moved from its current node in the kd-tree to the sibling node at a given tree level). Since point reassignments are expensive (requiring the retrieval of each reassigned point from the remote disk), this property is valuable for maintaining the partition under streaming updates to the dataset. 
Nonetheless, reassigning a single point at a single level of the kd-tree may lead to a cascade of reassignments in other levels. Fortunately, we prove (\Cref{thm:kd-tree-stability}) that the cascade of reassignments can lead to at most one reassignment per tree node.

To efficiently keep track of evolving coordinate medians at each kd-tree node over infinitely many sequential updates, we use priority queues. To this end, we develop our new I/O-efficient priority queues (\Cref{lem:max_DDS}). We define notions of insertion and deletion paths in heap trees, which facilitate batched disk-based updates by retrieving a constant number of entire root-to-leaf paths per update. Putting everything together, we prove the theoretical guarantees of our practical quantization method in \Cref{thm:tbq,thm:tbq_consistent}.

\subsection{Prior Methods}
There are several approaches for ANN quantization over a stream with minimal memory/disk communication. One is to ignore drift altogether and just stick with a quantization ``frozen'' on a prefix of the stream; we call this approach FrozenPQ.  This elides both challenges (i) and (ii). Similar in spirit is the recent RaBitQ method \cite{gao2024RaBitQ,gao2024practical}, even though it is not based on PQ. Its pre-processing step remains ``frozen'' on the original data and is no longer updated over the stream.
OnlinePQ \cite{online-pq} takes a different approach: it starts with a PQ partitioning and never changes that initial partitioning.  Instead, the quantization centers are shifted within a partition to accommodate insertions/deletions, which addresses challenge (i) but still abstains from challenge (ii).
The benefit of all these solutions is that no memory/disk communication is required: inserts and deletes are performed fully in memory.  However, their price is recall drop as the stream progresses, due to their non-adaptability or partial adaptability to data drift. 

On the other end of the spectrum, one can recompute the quantization from scratch every time the data changes, or periodically at fixed intervals.  
We call this method RebuildPQ. 
This method achieves the highest recall immediately after an index rebuild, although it is too slow, as it involves loading the entire dataset from disk to memory. Additionally, during stream intervals between rebuilds, recall degrades similarly to FrozenPQ. 

A potential middle ground is to allow limited disk/memory transfers and rely on them to mitigate recall degradation due to drift. The recent DeDrift algorithm \cite{dedrift} takes this approach. 
It continues to assign new points to existing centroids until clusters grow too large.  When that occurs, it reassigns points from the largest clusters to new clusters. 
DeDrift as originally presented in \cite{dedrift} was designed to update an inverted file index (IVF) over a stream. We use the term DeDriftPQ for its natural adaptation to the PQ paradignm, which yields a quantization method updateable over a stream with limited disk access. 
DeDriftPQ has lower latency than RebuildPQ due to its more frugal disk access, albeit still much larger than the disk-free methods FrozenPQ, RaBitQ and OnlinePQ.

An extended overview of related work is given in \Cref{sec:related-work}.

\section{Definitions and Computational Model}\label{sec:model}
Large-scale vector databases typically use a memory-disk architecture where quantized vectors reside in main memory and full-precision vectors are stored in external storage \cite{weber1998quantitative, ferhatosmanoglu2000vector, chen2021spann, jayaram2019diskann, matsui2018survey}. This design reflects the requirements that queries need low-latency access and storing large vectors exceeds practical memory budgets. Two modeling tenets characterize these workloads:
\begin{CompactItemize}
    \item Query processing should execute in main memory using compressed representations for approximate results, with optional re-ranking for refinement.
    \item Memory-disk communication should be minimized in both total data transfer volume and dependent I/O round-trips.
\end{CompactItemize}

\subsection{Dynamic Consistency}

We start by formalizing the notion of a dynamic data structure. Let $\mathcal{U}$ be a finite universe and let $\mathcal{F}$ be a family of query functions $f : 2^{\mathcal{U}} \to \mathcal{O}$ that map subsets of $\mathcal{U}$ to some output range $\mathcal{O}$.
A dynamic data structure $D$ supports the following operations:

\begin{CompactItemize}
    
\item[(i)] \textbf{Build:} Given $X\subset \mathcal{U}$, initialize $D$ with the subset $X$. 

\item[(ii)] \textbf{Query:} Given a query function $f \in \mathcal{F}$, return $f(X)$. 

\item[(iii)] \textbf{Insert:} Given $x\in \mathcal{U}$, add $x$ to $X$. 

\item[(iv)] \textbf{Delete:} Given $x\in X$, delete $x$ from $X$.
\end{CompactItemize}

Our goal is to develop dynamically consistent structures where after every update, query responses match those of a structure built from scratch. This ensures performance equivalent to optimal static quantization while supporting efficient updates. Dynamic consistency prevents gradual performance degradation under updates and eliminates prohibitively expensive rebuilds.

Formally, we view a dynamic index structure $D$ as maintaining an answer-state
$\mathrm{\mathbf S}_D \in \mathcal{O}^{|\mathcal{F}|}$
of its answers to all possible queries  $\mathcal{F}$. Denote by $\mathrm{\mathbf S}_D(\texttt{seq})$ its state after the sequence of operations $\texttt{seq}$. Note that query operations do not change the state. 

\begin{definition}\label{def:dynamicallyconsistent}
A deterministic dynamic index structure $D$ is \emph{dynamically consistent} if it satisfies for any $X \subset \mathcal{U}$, $y \in \mathcal{U} \setminus X$ and $x \in X$ the following: 
\begin{CompactItemize}
\item $\mathrm{\mathbf S}_D($`Build($X$); Insert($y$)'$)=\mathrm{\mathbf S}_D($`Build$(X\cup\{y\}')$,

\item $\mathrm{\mathbf S}_D($`Build($X$); Delete($x$)'$)=\mathrm{\mathbf S}_D($`Build$(X\setminus\{x\}')$. 
\end{CompactItemize}

If $D$ is randomized, it is \emph{dynamically consistent} if this holds for every fixed random seed.
\end{definition}

We remark that dynamic consistency is related to but different from the notions of weak/strong history independence (abbrev. WHI/SHI). These notions from the field of data structure security assume that an adversary observes the memory state of the data structure. They require that the data structure’s \emph{memory state} remains the same (or at least indistinguishable to the adversary) under inserts/deletes that end with the same set of elements in the data structure. In contrast, dynamic consistency requires the \emph{answer state} of the data structure to remain to same, i.e., its answers to all queries, while allowing its memory state to be different (as in general, there are multiple different memory states that yield the same set of answers to all queries).

\subsection{The DySk Model}

We introduce our DySk model, a memory-disk computational model tailored to capture dynamic ANN quantization. 

To motivate our need for a new model, let us review existing models in the literature in the context of ANN quantization. We recall the desiderata of ANN quantization. We have limited  main memory for storing the quantized data and a slower remote disk for storing the full data. 
The quantized dataset contains a sketched representation of all points, thus its size need be linear in the number of points, but smaller the total size of the unsketched dataset (due to its dimensionality and coordinate bit precision). 
ANN queries should be answerable from the quantized data using main memory only (with a possible subsequent disk I/O for refining the ranking of results). Streaming updates to the dataset may require unsketched data access and may probe the remote disk, but the size of the data transfer and the number of consecutive disk I/Os should be kept to a minimum.

With these desiderata in mind, we go over existing memory-cost models for data structures:
\begin{CompactItemize}
    \item The classical \emph{cell probe model} \cite{miltersen1999cell} is designed to capture memory access costs in data structures. However, it only models a single storage device rather than memory-disk interplay.
    \item The \emph{external memory model} \cite{vitter2001external} incorporates secondary storage access. However, it assumes a disk-centric design where typically every operation requires multiple consecutive I/Os. The main memory size is fixed to be independent of (and smaller than) the dataset size, thus it cannot store a quantized representation per point, and cannot support operations like ANN queries that may need to return any data point without disk access. 
    \item The \emph{sketching model} \cite{nelson2020sketching} is designed to enable query processing in constrained memory through lossy compression with a sketch size that scales with the dataset size. However, it provides no mechanism for disk-based maintenance of sketch consistency under data modifications. Streaming updates to the sketch rely solely on the sketched data (e.g., \cite{li2014turnstile,braverman2017clustering}).

\end{CompactItemize}

\noindent The DySk model adapts elements from these models to capture ANN quantization. 
Essentially, it can be seen as the sketching model endowed with external memory for sketch updates under insertions and deletions. 
It maintains and queries the quantized data in a main memory whose size scaled with the dataset size, and enables access to the full data on a remote disk for updates only, restricting each update to $O(1)$ dependent I/Os independently of dataset size. 
We do not explicitly parameterize the block size or the max transfer size per I/O, and instead treat it as a complexity measure to be minimized (in our data structures it will generally be polylogarithmic in the dataset size). 
Measuring both total data transfer sizes and dependent I/O roundtrips enables the model to handle different storage architectures, from local disks to cloud object stores and disaggregated systems, where I/O latency can vary and dominate throughput costs. The strict bound on consecutive I/Os avoids costly dependent I/O operations that would incur cascading latency penalties that compound across roundtrips.
 
\paragraph{Model definition.}
Formally, the DySk model is defined as follows. 

\noindent\textbf{Storage devices.}  We model two storage devices: main memory and external storage, which we refer to as ``disk''.
We measure main memory space in bits and disk space in machine words.
Main memory is free to access but has limited size.   
Disk is expensive to access but has large size. 
The disk is modeled as a table that maps an address $\texttt{a}$ to an arbitrarily sized piece of data $\texttt{data}(\texttt{a})$ stored at that address.
The main memory stores two pieces of information for each element: 
\begin{CompactItemize}
\item a local sketch $sk(i)$, which is a compressed synopsis of the point $x_i$; 
\item a disk address $\texttt{addr}(i)$ where the full data associated with $x_i$ is stored remotely.
\end{CompactItemize}

\noindent\textbf{Disk I/Os.} 
Query operations are considered performance-critical in the DySk model and must run fully in-memory, without disk access. Insert and delete operations may access the remote disk by two types of I/O -- reading and writing. A read I/O specifies a list of disk addresses $\texttt{A}_R$ and retrieves the collection $\{\texttt{data}(\texttt{a}):\texttt{a}\in\texttt{A}_R\}$ from disk to memory. 
A write I/O specifies a list of disk addresses $\texttt{A}_W$ and corresponding data $\{\texttt{data}^*(\texttt{a}):\texttt{a}\in\texttt{A}_W\}$, and overwrites those addresses in the disk as $\forall_{\texttt{a}\in\texttt{A}_W}\texttt{data}(\texttt{a})\leftarrow\texttt{data}^*(\texttt{a})$. 

\noindent\textbf{Efficiency measures.}
The primary measures are the main memory size, the number of sequential I/Os, and the total I/O size. 
Our goal is for the main memory to store a pre-specified number of $t$ bits per point ($\forall_i |sk(i)|=t$), corresponding to a desired quantization level, while supporting inserts/deletes with  $\log^{O(1)}(n)$ total I/O size and $O(1)$ sequential I/Os. 

The main memory size is defined as $\sum_{i=1}^n|sk(i)|$ over the $n$ points currently in the data structure. Note that this disregards the disk addresses $\texttt{addr}(i)$ which impose a fixed cost on all data structures. The size of a read I/O is $\sum_{\texttt{a}\in\texttt{A}_R}|\texttt{data(a)}|$, and the size of a write I/O is $\sum_{\texttt{a}\in\texttt{A}_W}|\texttt{data}^*(\texttt{a})|$. Both correspond to the total amount of data communicated between memory and disk. 

\paragraph{ANN quantization.}
The DySk model captures quantization in vector databases, storing $b$-bit quantizations $sk(i)=\widetilde x_i$ of each $x_i\in X$ in memory. At query time, a subset of $k'$ candidate vector IDs $I\subset[n]$ is formed in-memory as the $k'$ nearest $\widetilde x_i$'s to the query, where $k'$ satisfies $k<k'\ll n$ and governs the I/O size. Then a ``reranking'' read I/O is made to retrieve the full representation of the candidates $\{x_i\}_{i\in I}$ from the remote disk. The $k'$ candidates are reranked by their exact distance from the query, and the $k$ nearest neighbors among them are returned. Our goal is to endow this setup with streaming data updates. The memory-only requirement of the DySk model queries ensures that the quantized candidate selection stage does not introduce additional I/Os on top of the reranking I/O at query time.

\section{Provably Accurate ANN Quantization with Dynamic Updates}\label{sec:quadsketch}
To present our first result, we define the following formal notion of ANN. 
In the \emph{$(\epsilon,\delta)$-approximate nearest neighbor search} problem ($(\epsilon,\delta)$-ANN), for any query $q\in\mathbb{R}^d$, the data structure must return a vector $x\in X$ such that, over the internal randomness of the data structure,
\[ \Pr\left[\|q-x\|\leq(1+\epsilon)\min_{x^*\in X}\|q-x^*\|\right] \geq 1-\delta . \]

A long line of work has studied upper and lower bounds for dataset compression for ANN. 
The Johnson-Lindenstrauss dimensionality reduction lemma
\cite{johnson1984extensions,achlioptas2001database,ailon2006approximate,larsen2017optimality} is key to compression by reducing data dimension. Other works have studied further reduction in the total number of bits by sketching (i.e., quantizing) the dataset  \cite{kushilevitz1998efficient,alon2017optimal,indyk2017near,quadsketch,indyk2018approximate,indyk2022optimal,pagh2019space,huynh2020fast,zhang2021faster}. 
In particular, Indyk et al. \cite{quadsketch,indyk2018approximate} 
 showed that $(\epsilon,\delta)$-ANN can be solved with a quantization size of $O(\log(1/\epsilon))$ amortized bits per coordinate. 
However, these results are for static quantization and do not support point insertions and deletions. 

Here, we prove that this \emph{static} quantization result can be extended within the DySk model into a \emph{dynamically consistent} quantization result that formally solves the streaming $(\eps,\delta)$-ANN problem.  See \cref{app:quadsketch} for the formal definitions and proofs.

\begin{theorem}\label{thm:quadsketch}
There is a dynamically consistent data structure that solves the streaming $(\eps,\delta)$-ANN problem with the following guarantees:
    \begin{CompactItemize}
        \item \textbf{Memory size:} 
        $O(n\frac{\log N}{\eps^2}(\log(\frac{1}{\eps\delta}) + \log\log N))$ bits.
        \item \textbf{I/Os per update:} One read I/O and one write I/O, each of size $O(\eps^{-2}\log N)$ words.
        \item \textbf{Update and Query time:} $\tilde{O}(\eps^{-2}\log^2 N)$ each.
    \end{CompactItemize}
\end{theorem}
For reference, the memory size after dimension reduction and without further quantization would be $O(n\tfrac{\log N}{\epsilon^2}\log n)$, due to $O(\tfrac{\log N}{\epsilon^2})$ dimensions and $O(\log n)$ bits per dimension, assuming for simplicity that full coordinate precision fits in $O(\log n)$ bits per coordinate (we remove this assumption in the appendix). 
Thus, similarly to \cite{quadsketch,indyk2018approximate,indyk2022optimal}, our result improves the bound to a single log factor, which is necessary in static settings \cite{alon2003problems,molinaro2013beating,larsen2017optimality,indyk2022optimal}.

The challenge in this result is that the static methods of \cite{quadsketch,indyk2018approximate} heavily relies on lossy data-dependent compression and cannot support streaming updates without additional access to full-precision points.  
\Cref{thm:quadsketch} leverages the DySk model framework for showing that limited disk I/Os (of constant number and logarithmic size) suffice to that end. The proof is in \Cref{app:quadsketch}.

\section{Practical and Provable Method: \ouralg}\label{sec:tbq}

In this section we leverage the theoretical foundation developed in the preceding section to present a principled and practical method, \ouralg: Consistent Dynamic Efficient Quantizer, our solution for streaming quantization for vector databases.
\ouralg\ leverages data-dependence while supporting efficient updates via median-based partitions that remain stable under changes and need only limited reassignments. We prove that \ouralg\ is dynamically consistent with bounded I/O operations through disk-aware priority queues.
The theoretical guarantees for \ouralg\ are stated in \Cref{sec:codeq_theory} after the detailed presentation of the method.

\subsection{The Basic Quantizer}
\ouralg\ is based on kd-trees. However, the way it uses them is different from their classical use. While usually their role is reducing search time by pruning the dataset with a beam search, \ouralg~uses the kd-tree partition as a quantizer to reduce memory size. 

A kd-tree induces a hierarchical data-dependent partition of the dataset determined by coordinate medians at each tree node. 
For clarity, we briefly recall how to build a depth $L+1$ kd-tree. Let $X \subset \mathbb{R}^d$ be a dataset.
First, uniformly sample a random sequence of indices $(j_1, \dots, j_{L}) \in \{1, \dots, d\}^L$ without replacement.
The tree root of the kd-tree is associated with all points in $X$. We split it into two halves by at the median according to coordinate $j_1$. 
The left child of the root inherits all points less than the median in coordinate $j_1$ and the right child inherits the remaining points.
The splitting procedure continues recursively on the child nodes, wherein  nodes at level $\ell$ use the coordinate $j_\ell$ for their median split. 
The $2^L$ leaves of the kd-tree then define a partition of $X$ into $2^L$ clusters.

In \ouralg, we view these clusters as proxies for $k$-mean clusters (with $k=2^L$) and compute the mean of each cluster to serve as a proxy for the points in that cluster. 
The quantized representation of each point is the $L$-bit ID of the leaf it resides in. 
To improve effectiveness in high dimensions, we use the kd-tree quantizer as the subvector quantizer in each block of a product structure. This results in a PQ-like quantization scheme, except that $k$-means quantization in each block is replaced by a kd-tree quantization with the same number of clusters.  
This is depicted in \Cref{pq-tbq}. 

\begin{figure}[t]
\centering
\includegraphics[width=0.7\linewidth]{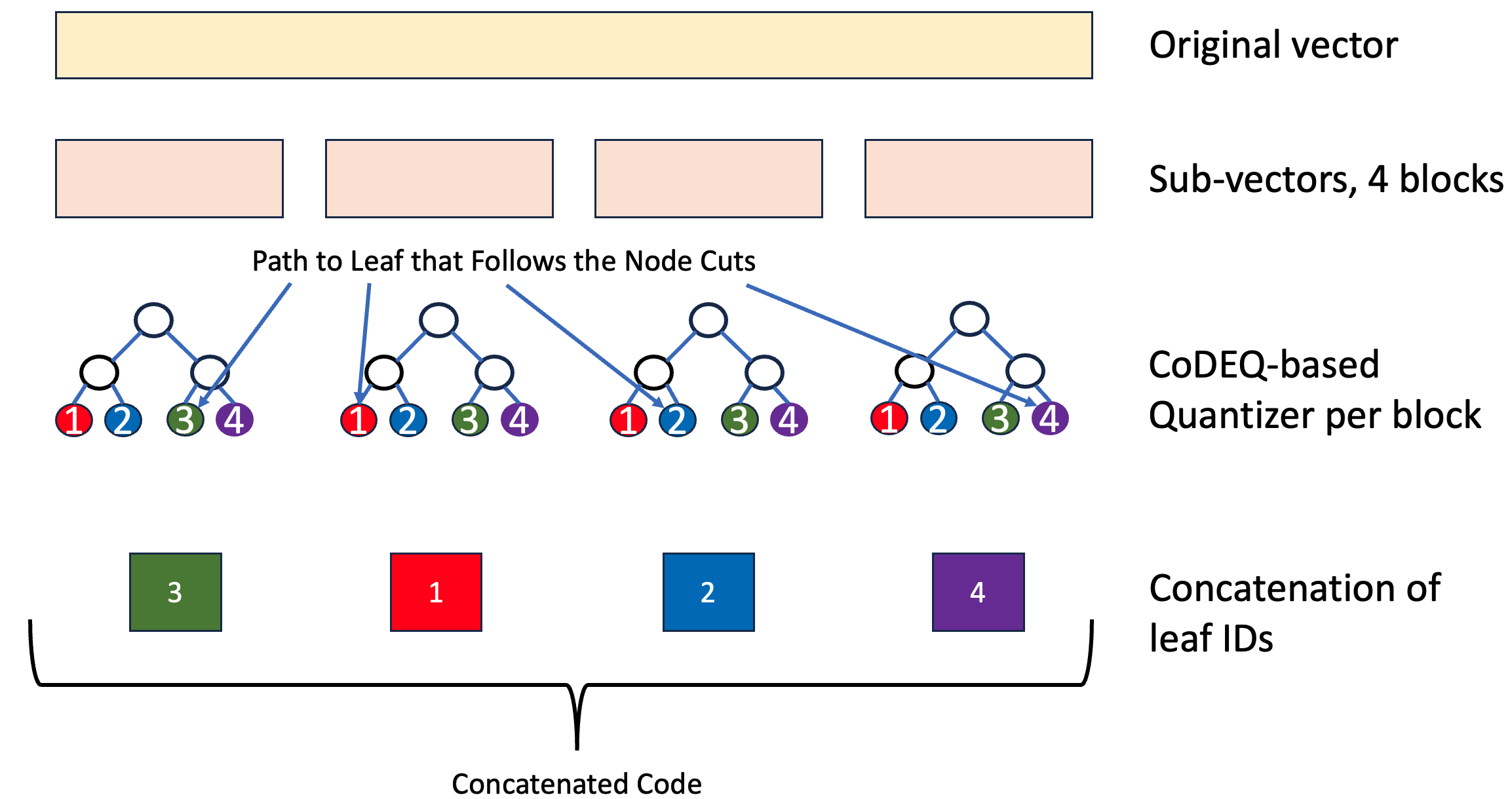}
\caption{Illustration of a \ouralg\ based Product Quantizer. The top row depicts a $d$-dimensional vector. 
Each vector is partitioned into blocks, shown on the next row as sub-vectors.  A \ouralg\ is computed on each block of dimensions.
This \ouralg~substitutes the $k$-means algorithm used in the usual PQ algorithm.  The leaves of each tree represent the codebook. 
The last row shows how the codebook is computed, i.e., concatenating the leaf node IDs corresponding to which leaf each subvector falls into in the tree.}
\label{pq-tbq}
\end{figure}

\subsection{Dynamic Updates}

\begin{wrapfigure}
{R}{0.6\textwidth}
\includegraphics[width=\linewidth]{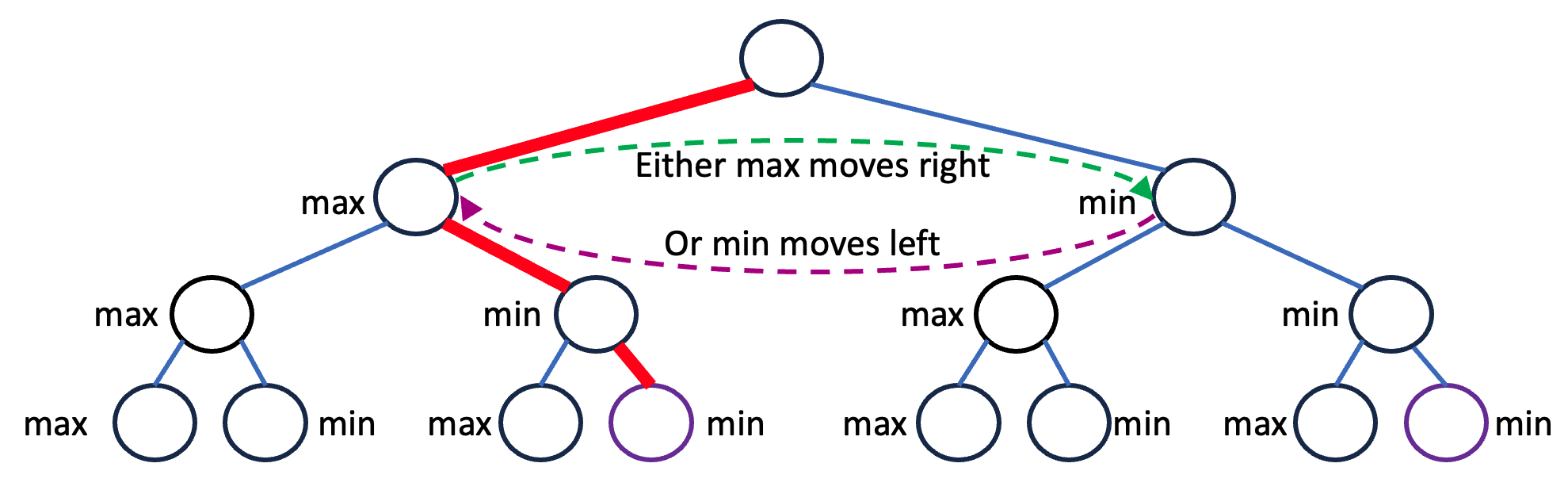}
\caption{
Illustration of a \ouralg~update. 
As the new point is inserted/deleted from each node along the path to its leaf, the median of the node can shift by one position, requiring reassigning some point from one child to the other. This reassignment can lead to further cascading reassignments \Cref{thm:kd-tree-stability} proves that the cascading effect can lead to at most one point reassignment per tree node.}
\label{update-tbq}
\end{wrapfigure}

Now we arrive at our motivation for replacing $k$-means with kd-trees: $k$-means clustering cannot be updated over a stream without potentially having to retrieve many full-precision points from the remote disk. In contrast, we show that the kd-tree quantizer makes this possible. 

The crucial property we exploit is that the median-based partition at each kd-tree node is \textbf{stable} under incremental updates. If a point is inserted or deleted, the coordinate-median of the dataset moves by at most one position to the left or to the right. Therefore, only one point may need to be reassigned to the other side of the partition. 
Furthermore, the identity of the point that needs to be reassigned can be easily determined: it must be one of the two points on the partition boundary.
Since re-assigned points need to be retrieved from disk, this stability property is useful in minimizing disk access while supporting updates: there are only two points to fetch from disk, and computing which points these are is immediate and causes no latency. 

Leveraging this property for an efficient streaming quantizer raises a few challenges. One challenge is the potential cascading effects of point reassignment. 
When a new point is inserted or deleted from the kd-tree, its inserted/deleted from all of the nodes along  the root-to-leaf path to the leaf that contains it leaf, as depicted in \Cref{update-tbq}. This already entails a potential point reassignment per tree level. But these are not all the necessary reassignments. When a point is reassigned from (say) the bottom half to the top half of the partition in tree node, it gets moved from its left-child to its right-child, meaning it causes an insertion and a deletion in the child nodes. These in turn may cascade to further insertions and deletions in yet more nodes. Can this cascade of changes in the tree be bounded?
Fortunately, the answer is yes. We prove that this cascading effect can lead to at most one reassignment per tree node. The proof of the following theorem is in \Cref{app:tbq_stability}.

\begin{theorem}\label{thm:kd-tree-stability}
Fix a dataset $X \subset \mathbb{R}^d$ and integer $L > 0$. 
Let $T$ be a depth $L$ kd-tree built on $X$. 
Let $X'$ be the dataset after inserting or a deleting a point.
Then, every node of the new kd-tree $T'$ has changed by removing at most one point from and/or inserting at most one point into the corresponding node in $T$.
Furthermore, the only points $X$ that can change their root-to-leaf path are the maximum (resp.~minimum) element of a node $v$ which is a left (resp.~right) child in $T$.
\end{theorem}

The second challenge with the stability of median partitions is that it addresses only a single update to the quantizer. At the beginning of the stream we know that the two points that might need to be fetched from disk are those adjacent on the partition boundary. 
Once a point reassignment occurs, there would be new points adjacent on the boundary. 
They might get reassigned at the next streaming update, introducing other points at the decision boundary, and so on. To maintain efficient quantizer updates over an infinite stream of inserts and deletes, we need a way to keep track of which points are adjacent on the evolving partition boundary of each kd-tree node.

Since the two points adjacent on the partition boundary are always the maximum of the bottom half and the minimum of the top half (where the maximum and minimum are w.r.t. the coordinate $j_\ell$ by which that kd-tree node partitions the data points), the natural way to keep track of them under data updates is with priority queues. Indeed, a priority queue precisely keeps track of the minimum/maximum element in a set under insertions and deletions. 
Nonetheless, we cannot just use vanilla priority queues, as they would occupy too much space on the main memory. Rather, we need to design memory/disk efficient priority queue. We do so in the next section.

\begin{wrapfigure}
{R}{0.4\textwidth}
\begin{center}
\includegraphics[width=0.4\textwidth]{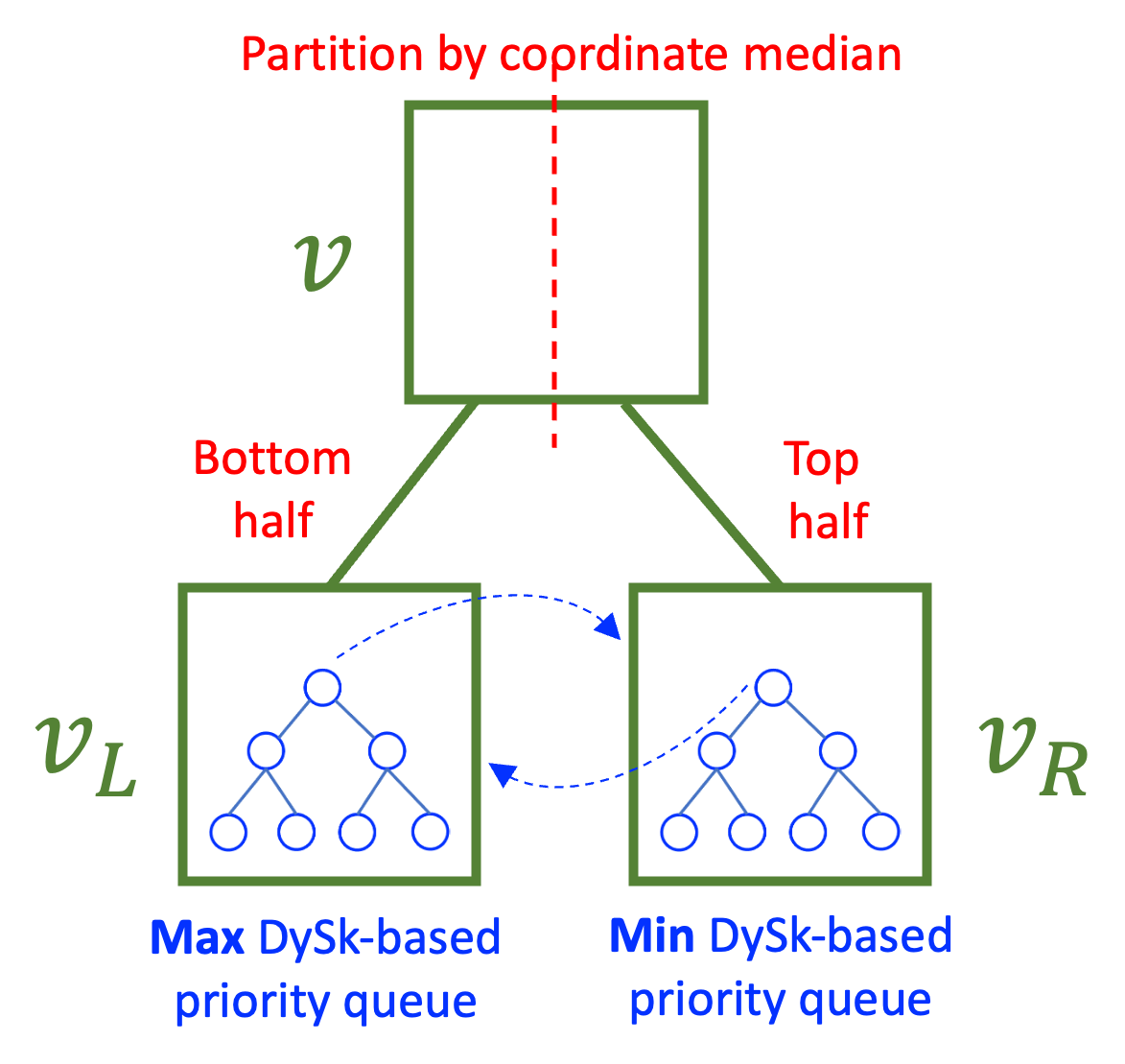}
\caption{\ouralg\ single node update. A parent node $v$ in the kd-tree partitions its cluster along a coordinate median. The left and right children $v_L,v_R$ maintain a DySk-based max and min priority queue, respectively.
Upon insertion/deletion to $v$, the coordinate median changes, and either the maximum element in $v_L$ moves to $v_R$, or the minimum element in $v_R$ moves to $v_L$. This may cause further cascading updates in other nodes (see \Cref{{update-tbq}}).
}
\label{fig:kdtreesmall}
\end{center}
\end{wrapfigure}

\subsection{Disk-Efficient Priority Queues}\label{sec:tbq_priority}

A priority queue  is a classical data structure for maintaining the maximum\footnote{We describe a maximum priority for concreteness. A minimum priority queue is similar.} priority element in a set of points under data updates. 
Formally, it maintains a set of elements $X = \{x_1, \dots, x_n\} \subset U$ from a strict totally ordered universe $U$.
The goal of the data structure is to answer a single query: return $(x^*, i^*) = (\max_{i} x_i, \argmax_{i} x_i)$.

For ANN quantization we need a disk/memory efficient priority queue. This means that \emph{(i)} the priority queue should minimize main memory utilization, maintaining most of its data on the remote disk, and \emph{(ii)} data updates should be disk efficient, minimizing I/O dependency and size.

Existing priority queue approaches in external memory models (e.g., \cite{kumar1996improved,fadel1999heaps,eenberg2017decreasekeys,jiang2019faster}) typically require many dependent I/Os. This makes them unsuitable for streaming quantization where minimizing dependent I/O roundtrips is critical for update efficiency.
Loosely described, existing priority queues maintain a binary heap tree on disk and its root (the maximum or minimum element) in the main memory. When a point is inserted or deleted, they traverse up or down the tree to propagate the update, fetching from the remote disk the appropriate node on the traversal path. In many block size regimes, this leads to a logarithmic number of dependent I/Os -- each I/O needs the previous one to complete before it can begin -- which causes prohibitive latency in modern ANN settings.

As an alternative, we develop memory/disk efficient priority queues with a constant number of dependent I/Os per update. The main result for this section is stated next and proved in \Cref{app:tbq_priority}.

\begin{theorem}\label{lem:max_DDS}
There is a dynamically consistent priority queue with the following properties: 
    \begin{CompactItemize}
        \item \textbf{Main memory space:} $b + \log N$ bits of space, where $b$ is the number of bits needed to store a single numerical value.
        \item \textbf{I/Os per-update:}
        A sequence of three read I/Os of size $O(T)$, $ O(T)$, $O(T\log n)$ words and a single write I/O of size $O(T\log n)$ words where $T = \log N$.
        \item \textbf{Update and Query time:} $O(\log n\log N)$ and $O(1)$, respectively.
    \end{CompactItemize}
\end{theorem}

Here we give an overview of the construction and proof. 
Our priority queue maintains the maximum element in main memory, so queries can be answered in $O(1)$ time without disk access. On the disk, we store a maximum heap tree, which we recall is a balanced binary tree with a one-to-one correspondence between nodes and values in $X$, such that the value at each node is larger than the values at both its children. 

To bound the I/O cost of insertions and deletions we define the notion of \emph{insertion/deletion paths}. The \emph{insertion path} $P_+$ of the heap runs from the root to the first vacant leaf. For each tree node $v$, its \emph{deletion path} $P_v$ runs from the root to $v$, then proceeds downward to the larger of the two children, until reaching a leaf. 

Insertion and deletion paths have the following useful properties: 
\emph{(i)} It can be observed that insertion only impacts the insertion path $P_+$ and the deletion paths $\{P_u\}$ of the nodes $u$ that reside on $P_+$. Similarly, the deletion of a value stored in a node $v$ only impacts $P_+$, $P_v$, and the deletion paths $\{P_u\}$ of nodes $u$ that reside on $P_+$ or $P_v$. 
\emph{(ii)} Crucially, the impact on an impacted insertion/deletion path takes a clear and convenient form: the path is diverted at some node to its sibling, then proceeds with the sibling's deletion path. Therefore, updating an impacted insertion/deletion path boils down to copying over a sibling path suffix to override the current path suffix. 

Due to these structural properties, in order to update the priority queue, we do not need to traverse the heap tree level by level with many consecutive I/Os. Instead, we can fetch from the disk at once the insertions and deletion paths that might be impacted by the update (those mentioned above), and their siblings' deletion paths in order to effectuate the necessary updates by copying over sibling path suffixes as needed. 

As we know in advance all the insertion and deletion paths we might need for the update, we can fetch them from the disk with $O(1)$ consecutive I/Os. Furthermore, the total size of those paths is polylogarithmic, and hence those disk I/Os are of small size. 
The process is depicted (for deletion) in \Cref{fig:heapdel}, with full formal details provided in \Cref{app:tbq_priority}.

\begin{figure*}[t]
\begin{center}
    {\includegraphics[width=0.12\linewidth]{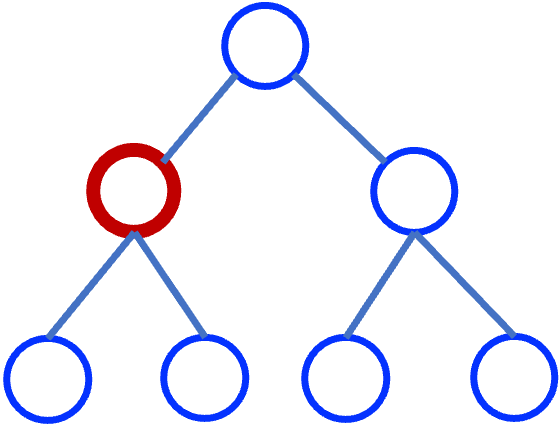}}
\hfil
    {\includegraphics[width=0.12\linewidth]{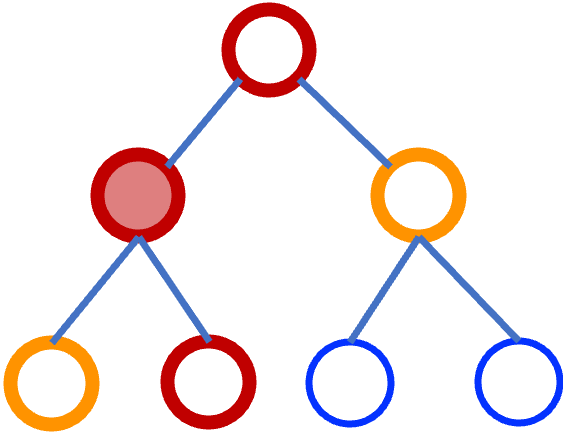}}\hfil
    {\includegraphics[width=0.12\linewidth]{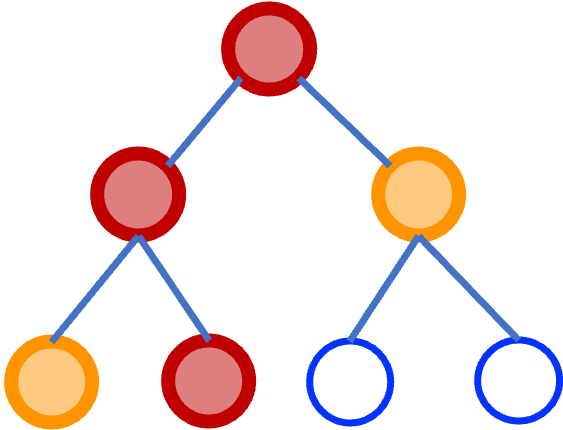}}\hfil
    {\includegraphics[width=0.12\linewidth]{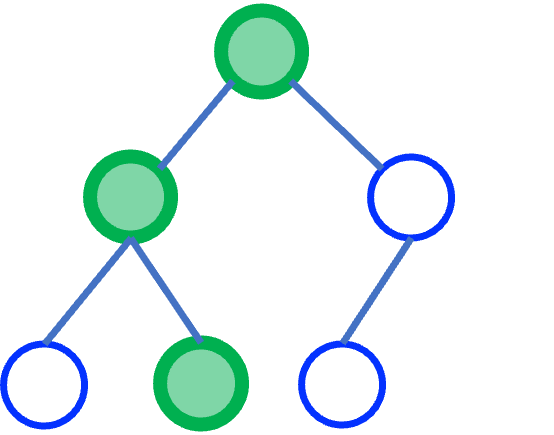}}
\caption{Heap deletion I/Os in the DySk model. Let $x_i$ be the point to be deleted. Disk access, left to right: \textbf{1st I/O:} read from $\texttt{addr}(i)$ the \textcolor{BrickRed}{node address} $\texttt{a}_v$ in which $x_i$ is currently stored in the heap. \textbf{2nd I/O:} read from $\texttt{a}_v$ the addresses $\{\texttt{a}_u\}$ of the \textcolor{BrickRed}{nodes on the deletion path} $P_v$ and the \textcolor{orange}{nodes sibling to them}. There are $\leq2\log n-1$ such nodes. \textbf{3rd I/O:} read the contents of all addresses $\{\texttt{a}_u\}$ retrieved in the 2nd I/O. \textbf{4th I/O:} \textcolor{Green}{Write back to disk} the updated node contents and deletion paths.}
\label{fig:heapdel}
\end{center}
\end{figure*}

\subsection{Putting Everything Together}

We now give the full \ouralg\ algorithm. Pseudocode is included in the appendix. 

As a preprocessing step, we perform a standard (data-oblivious) random rotation of the data. Since a kd-tree partitions the data points across randomly chosen coordinates, a random rotation helps neutralize differences in variance between them.\footnote{See \Cref{app:rr} for details on the oblivious random rotation step.} The resulting \ouralg\ build algorithm is:

\begin{CompactEnumerate}
\item Randomly rotate the dataset $X$ to obtain a dataset $X'$.
\item Build a depth $L+1$ kd-tree on $X'$.
\item For each node $v_l$ that is a left child, store a DySk based max priority queue on the points associated with $v_l$. 
Similarly, for each node $v_r$ that is a right child, store a DySk based min priority queue on the points associated with $v_r$. This is depicted in Figure \ref{fig:kdtreesmall}. 
\item For each leaf $v$, compute the mean $\mu_v$ of all points in $X$ (un-rotated) that land in $v$ after rotation.
\end{CompactEnumerate}
The leaf mean $\mu_v$ will serve as the quantized proxy for distance queries for every $x_i$ that lands in the leaf $v$.
The memory-resident data in~\ouralg~is the $L$-bit encoding per point, and a ``codebook'' of size $2^L d\bits$ to store all leaf averages, for a total of $nL + 2^L d \bits$, as standard in ANN quantization. Note that $2^L$ is the number of clusters (the analog of $k$ in $k$-means), so notwithstanding the exponential notation $2^L$, the memory size is the same as PQ for the same quantization resolution. 

The algorithm for inserting or deleting a point $x$ is as follows:
\begin{CompactEnumerate}
    \item Insert $x$ to or delete $x$ from every node on its root-to-leaf path in the kd-tree. 
    \item A node $v$ in the kd-tree is called \emph{unbalanced} if the number of points in its left child differs from the number of points in its right child by more than one. The foregoing insertions/deletions may cause some tree nodes to become unbalanced.
    \item While there is an unbalanced node $v$ in the kd-tree, let $\ell$ be its level and $j_\ell$ the coordinate it splits by. Rebalance $v$ by either reassigning the largest point from the left child to the right child, or the smallest point from thr right side to the left child, where the ordering of the points by coordinate $j_\ell$. Reassigning a point requires reading its full precision from the disk, since its new location in the tree after reassignment may depend any of its coordinates. 
\end{CompactEnumerate}

Note that inserting/deleting a point from a tree node $v$ means also inserting/deletion it from the min/max heap stored in $v$.

\subsection{Theoretical Guarantees}\label{sec:codeq_theory}
We recall notation and introduce some new notation:
\begin{CompactItemize}
    \item $n$: The number of points currently in the data structure.
    \item $d$: Data dimension.
    \item $N$: Upper bound on the number of points in the data structure at any point of its lifetime.
    \item $b$: Number of bits required to store a single numerical value in the underlying system.
    \item $L$: A user-chosen parameter that determines the quantization resolution. 
\end{CompactItemize}
each data point would be quantized into an $L$-bit representation. Note that $2^L$ is the analog of $k$ in $k$-means, thus a $2^L$ term is linear in the number of clusters.

All full proofs from this section are in \Cref{app:tbq}. 
Our first result is the efficiency of \ouralg.

\begin{theorem}\label{thm:tbq}
\ouralg\ can be implemented with: 
\begin{CompactItemize}
        \item \textbf{Main memory space:} $nL +  2^L d \bits + 2^{L+1}(2\log N + d\bits)+ d^2b$ bits of space.
        \item \textbf{I/Os per-update:}
        A sequence of three read I/Os of size $O(T)$, $ O(T)$, $O(T\log n)$ words and a single write I/O of size $O(T\log n)$ words where $T=2^L\max\{d,L, \log N\}$.
        \item \textbf{Update time:} $O(d^2 + 2^L \log n \max\{d, L, \log N\})$ time.
    \end{CompactItemize}
\end{theorem}
\begin{proof}[Proof sketch]
    We start with main memory size. Each point is quantized into $L$ bits, specifying the ID of the kd-tree leaf (i.e., the cluster) with which it is associated. This leads to the $nL$ term. The codebook contains the mean of each cluster, which is the point that serves as the quantized proxy for points in that cluster. Since there are $2^L$ clusters, and each mean takes $db$ bits to solve, this leads to the $2^Ldb$ term. By Lemma \ref{lem:max_DDS}, each priority queue takes $\log N+b$ bits in the main memory. Since the root point in the priority queue might get reassigned to other kd-tree nodes during an update, we keep its full representation (size $db$) in memory, as it is a negligible increase (dominated by the codebook size) and saves a disk I/O during updates. Since we have a priority queue per internal kd-tree node, this leads to the $2^{L+1}(\log N+db)$ term. Finally, storing the random rotation matrix takes $d^2b$ bits. This accounts for main memory usage. The I/Os per update and update time are inherited from the priority queues in Lemma \ref{lem:max_DDS} and the stability bound in Theorem \ref{thm:kd-tree-stability}. 
\end{proof}
We note the following important features of this theorem:
\begin{CompactItemize}
    \item The main memory term is governed by $nL$, i.e., $L$ quantized bits per point. The other low-order terms are for storing the codebook and additional auxiliary information (similarly to PQ and most other quantization methods).
    \item Disk access uses only $O(1)$ consecutive I/Os, each of size only polylogarithmic in $N$.
    \item Update time is only polylogarithmic in $N$. 
\end{CompactItemize}
We remark that Theorem \ref{thm:tbq} states the the efficiency bounds for a single \ouralg\ quantizer (one kd-tree). In a product \ouralg\ scheme with $m$ subvectors, we of course maintain $m$ copies of the above theorem, each with dimension $d/m$. 

Our second result is that \ouralg\ is dynamically consistent as defined in Definition \ref{def:dynamicallyconsistent}.

\begin{theorem}\label{thm:tbq_consistent}
    \ouralg\ satisfies dynamic consistency.
\end{theorem}
\begin{proof}[Proof sketch]
    The re-balancing of the median partition in all kd-tree nodes upon each \ouralg\ update guarantee that the kd-tree remains dynamically consistent. The dynamic consistency of the priority queues is by Lemma \ref{lem:max_DDS}. The mean of each cluster in the codebook is trivial to update and keep dynamically consistent upon updates (without any disk access required) since the mean is a linear function. 
    These are all the updateable components in \ouralg.
\end{proof}
\section{Experiments}\label{sec:experiments}

\begin{table}[hb]
\centering
\caption{Datasets and quantization settings}
\begin{tabular}{lrrrr}
\toprule
\multicolumn{3}{c}{\emph{Datasets}} & \multicolumn{2}{c}{\emph{Quantization}}\\
\cmidrule(lr){1-3} \cmidrule(lr){4-5}
\textbf{Source} & \textbf{Size} & \textbf{Dim.} & \textbf{Blocks} & \textbf{Bits/dim.} \\
\midrule
DEEP \cite{deep1b} & 100,000,000 & 96 & 8 & 1.0 \\
BigANN \cite{bigann} & 100,000,000 & 128 & 8 & 0.75\\
Text2Image \cite{deep1b} & 100,000,000 & 200 & 10 & 0.6\\
\bottomrule
\label{tab:data}
\end{tabular}
\end{table}

We empirically study the performance of product \ouralg\ in terms of recall performance under drift and disk I/O costs at varying scales. We first define a simple framework for constructing drifting vector search workloads from static benchmarks. We generate drifting scenarios for three popular benchmark datasets and evaluate the recall performance of \ouralg\ against several baselines. Finally, we compare disk communication costs between \ouralg\ and heuristic alternatives for PQ.

\paragraph{Methods.} We evaluate \ouralg\ against three categories of baseline: data-oblivious quantization, data-dependent but static quantization, and data-dependent dynamic quantization. We use LSHForest \cite{bawa2005lsh} for the oblivious baseline and product quantization (PQ) trained only on the initial dataset (``FrozenPQ") as the static baseline. Disk-dependent methods perform an update that requires previously seen full-precision data and thus communication with disk. We adapt a recent proposal for streaming IVF, DeDrift \cite{dedrift}, to the PQ setting by applying its heuristic k-means update to each of the PQ blocks. The DeDrift update assigns new data to existing clusters, then re-clusters all vectors belonging to the largest $m$ clusters by membership, where $m$ is small relative to $k$.

\paragraph{Metrics.} Our main metric for search performance is 
recall-$10@10$ (the fraction of true top-$10$ nearest neighbors also ranked in the top-$10$ by the quantizer).
This captures our specific focus on approximate nearest neighbor search while remaining agnostic to the choice of search algorithm that is combined with quantization in practice. To measure disk I/O cost, we count the number of full-precision vectors required from disk to perform the in-memory update step of a given quantizer. Since all in-memory quantizers will at least require writing new batches to disk, it is the read probe size that determines the difference in disk communication cost across various update methods.

\begin{wrapfigure}
{R}{0.45\textwidth}
\begin{center}
\includegraphics[width=0.45\textwidth]{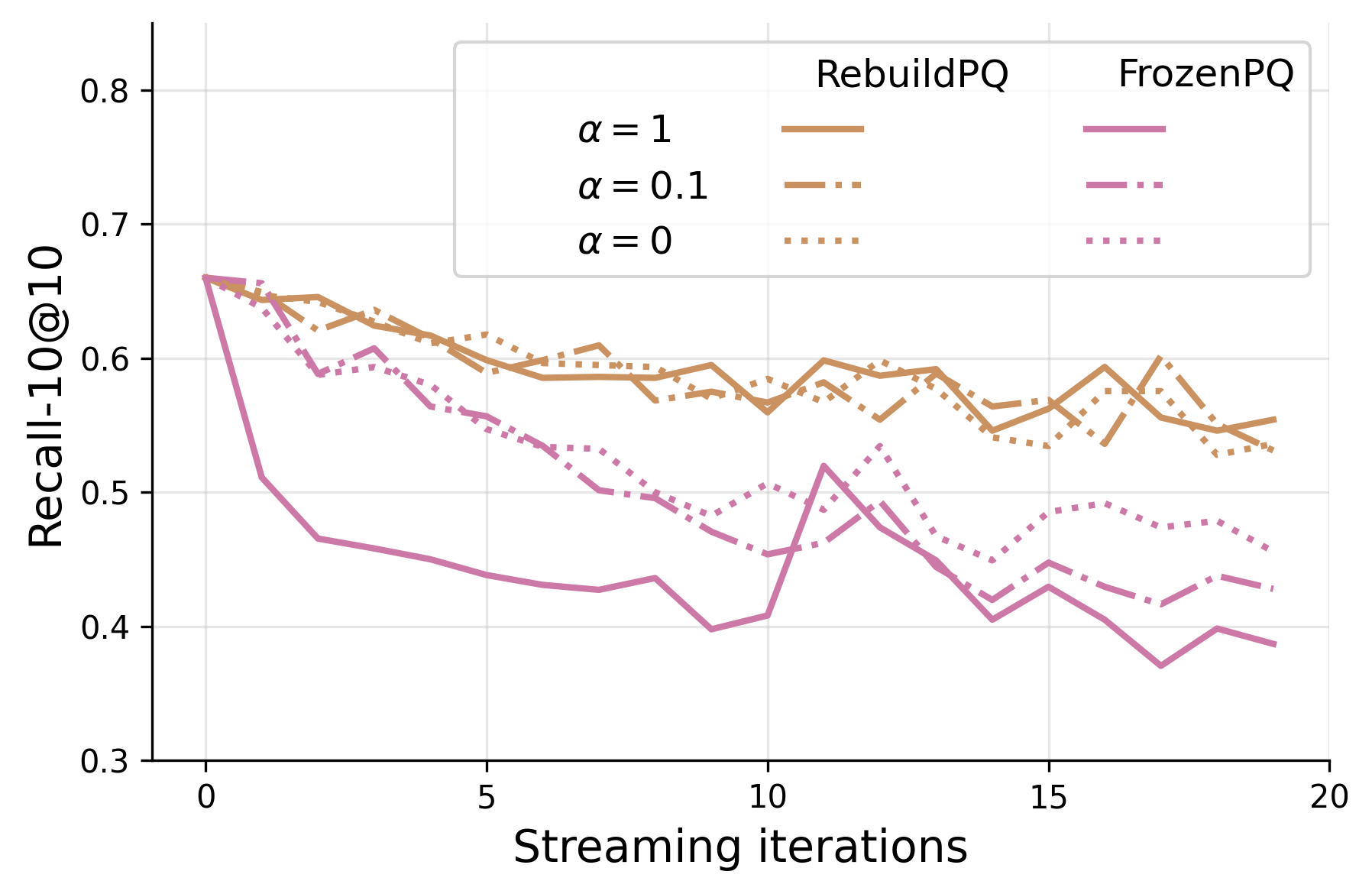}
\caption{Recall of static (“FrozenPQ”) vs.~fully-retrained (“RebuildPQ”) quantizers for the Deep-\emph{drift} dataset. The gap in performance depends on both data drift and query freshness, parameterized by $\alpha$. As $\alpha$ rises (fresher queries), the gap grows.}
\label{fig:freshness}
\end{center}
\end{wrapfigure}

\paragraph{Construction of streaming datasets.}

A streaming vector search workload can be modeled as a sequence of (update, query) tuples  $\{(U_t, Q_t)\}_{t=0}^T$, where at each time $t$ we apply the update $U_t$ to an existing dataset $X_{t-1}$, then search the updated dataset $X_t$ for the nearest neighbors of each element of $Q_t$. A valid update consists of inserts $I_t \subset \mathbb{R}^d \setminus X_{t-1}$ and deletes $D_t \subseteq X_{t-1}$. In each iteration, we obtain $X_{t+1}$ as $X_{t+1} = (X_t \setminus D_t) \cup I_t$.

Despite the practical interest of evaluating vector search in the streaming setting, there are few options for benchmark datasets with an intrinsic notion of streaming order and a relevant level of drift. Therefore we define a simple protocol to construct such a scenario from any static benchmark for vector search, similar to recent proposals in \cite{aguerrebere2024locally, mohoney2024incremental,simhadri2024results}. The full dataset is clustered and $X_0$ is instantiated as a subset of clusters. Inserts are drawn from the remainder. After each update, we search a new set of queries. The \emph{freshness} of these queries, i.e. their relative emphasis on the most recently inserted data, is parameterized by a scalar $\alpha$; $\alpha=1$ corresponds to queries drawn entirely from the cluster inserted at time $t$, while $\alpha=0$ corresponds to uniform sampling over all clusters observed so far. See Algorithm~\ref{alg:cdss} in~\Cref{app:experiments} for full details.

This protocol generates a streaming search dataset where the distribution of the underlying dataset gradually shifts over time, requiring a dynamic solution for quantization. In Figure~\ref{fig:freshness}, we validate that this is the case by comparing FrozenPQ to a fully-retrained PQ quantizer (“RebuildPQ”) over 20 iterations on a streaming dataset constructed from a 100K sample of Deep1B \cite{deep1b}. At time $t=0$ the quantizers are identical, yet over subsequent iterations a gap in recall performance emerges as the FrozenPQ quantization becomes increasingly stale with respect to the new data and query vectors. Moreover, the recall performance of the static quantizer decays more rapidly when queries are drawn from more recent data ($\alpha>0$). 
\paragraph{Recall under data and query drift.}

 We evaluate the recall performance of \ouralg\ on streaming datasets constructed according to the protocol above. We use 100M-scale versions of three ANN benchmark datasets: Deep1B \cite{deep1b}, BigANN \cite{bigann}, and Text2Image \cite{deep1b}. For each, we target quantization at or below 1 bit per dimension (see Table~\ref{tab:data}), reflecting the typical compression range of product quantizers such as PQ or CoDEQ. Parameters for streaming dataset construction are detailed in \cref{app:experiments}. To distinguish from their static sources, we refer to the streaming versions as Deep-100M-\emph{drift}, BigANN-100M-\emph{drift}, and Text2Image-100M-\emph{drift}, respectively.

 \begin{figure*}[t]
\begin{center}
\includegraphics[width=\linewidth]{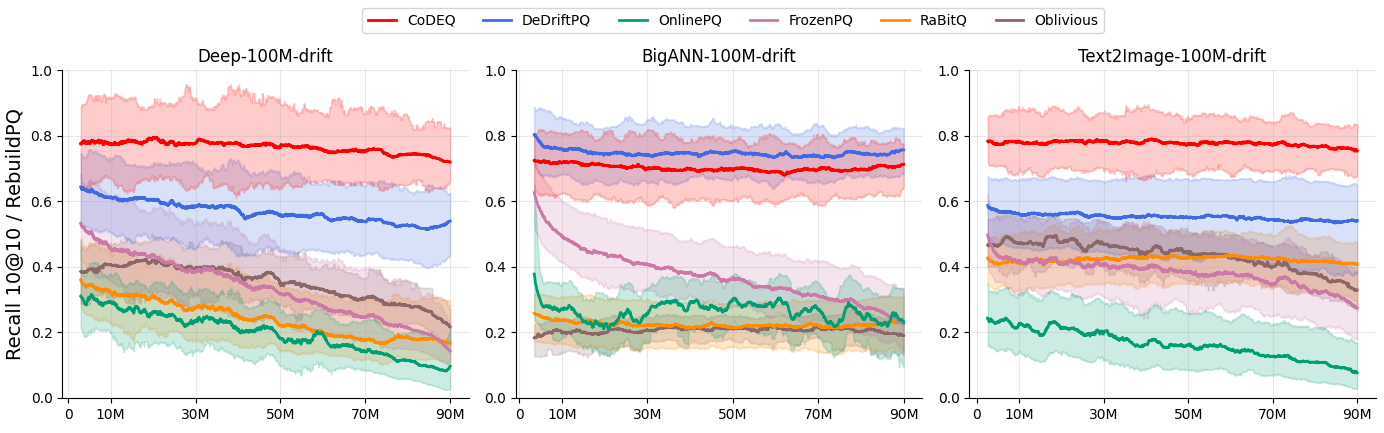}
\caption{Recall vs data streamed for the 100M-scale dynamic search scenarios. Solid lines denote the rolling median while shaded regions capture the rolling 10th-90th quantiles. Recall-10@10 is expressed as a fraction of RebuildPQ.}
\label{fig:recall}
\end{center}
\vskip -0.2in
\end{figure*}

Figure~\ref{fig:recall} reports the recall performance of \ouralg\ and baselines versus the total number of vectors streamed. By construction, the data and query vector distributions continue to drift from their initial states as the total number of vectors streamed increases. We show recall-10@10 for each quantizer expressed as a fraction of RebuildPQ, which represents an idealized but practically infeasible algorithm for streaming quantization. We plot the rolling median and shade a region representing the rolling 10th-90th quantiles. 

We observe two key features of \ouralg\ performance. First, it shows consistently superior recall over the baseline algorithms, with the exception of DeDriftPQ on BigANN-100M-drift, where the performance is comparable. Second, \ouralg\ maintains a roughly constant level of recall, despite the drifting updates and queries. By contrast, heuristic PQ updates such as OnlinePQ degrade versus RebuildPQ due to lack of dynamic consistency, while static or scalar quantization baselines are uncompetitive.

\begin{wrapfigure}
{R}{0.45\textwidth}
\begin{center}
\includegraphics[width=0.45\textwidth]{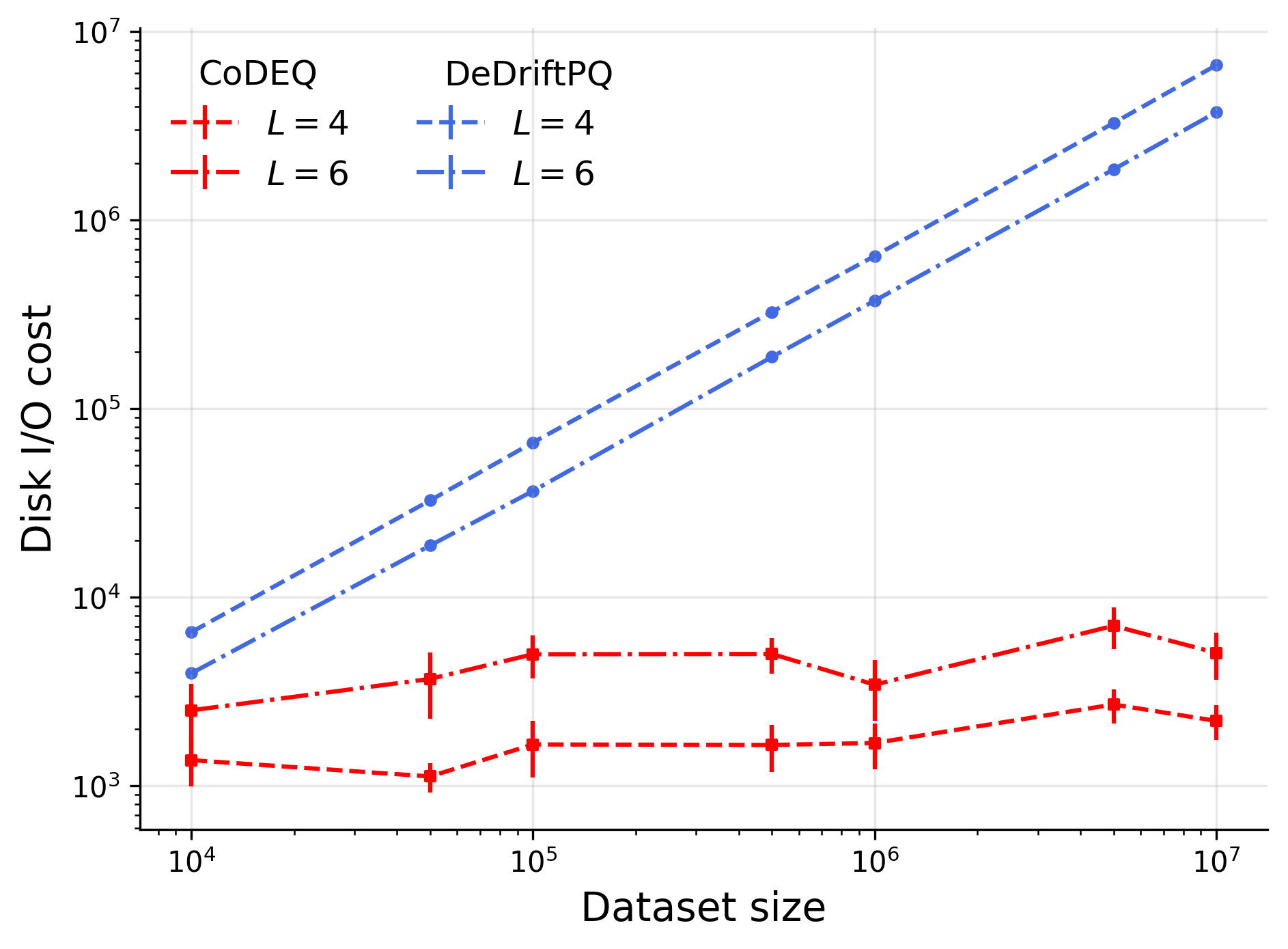}%
\caption{Disk I/O cost of a streaming update vs dataset size. Units of cost are full-precision vectors read from disk for a streaming update. Bars indicate std. error over 10 trials. }
\label{fig:iocosts}
\end{center}
\end{wrapfigure}

\paragraph{I/O costs of dynamic updates.} \ouralg\ is designed specifically for efficient communication with disk during updates. We show this empirically by measuring disk I/O cost for a single update to the quantizer and comparing directly to DeDriftPQ, the nearest competitor in terms of quantization recall performance under drift. Disk-independent baselines such as OnlinePQ, FrozenPQ, and RaBitQ have no disk I/O cost, but show substantially worse recall performance.

We sample datasets ranging over several orders of magnitude from Deep1B. For each, we initialize the quantizer on the full dataset, then compute the number of full-precision vectors required to update the quantizer when a single additional insert is made. We vary $L \in \{4, 6\}$. We set $m$ for DeDrift as a function of the $L$, specifically the smallest integer such that at least $2\%$ of clusters are reassigned. Note that in this context, smaller $m$ favors DeDrift, since this requires fewer clusters to be retrieved from disk. Results are shown in Figure~\ref{fig:iocosts}. The contrast between I/O costs illustrates the advantage of grounding \ouralg\ in the DySk model, with a hard limit on disk I/Os per update, while the DeDrift heuristic may require that a substantial fraction of the total dataset is retrieved.

\section{Extended Related Work}
\label{sec:related-work}

\paragraph{Quantization for ANN search.}
Product Quantization (PQ) \cite{vqindex2002, jegou2010product} and variants \cite{OPQ, gong2012iterative, babenko2014additive, kalantidis2014locally, matsui2018survey} are widely used for ANN search in vector databases.  
The $k$-means step in PQ scales poorly with the size and dimensionality of the data, with slow convergence in high dimensions.  
Our tree-based \ouralg\ approach replaces clustering with a dynamically consistent partitioning based on kd-trees and localizes changes to a limited number of vectors. 
Instead of maintaining explicit kd-tree partitions, we use them to quantize and represent vectors as compact codes. Quantizing the vectors in the same format as other methods simplifies incorporation of our approach into current similarity search libraries and services, such as FAISS, Milvus, NGT, PGVector, and OpenSearch  \cite{douze2024faiss, wang2021milvus, ngt, pgvector, opensearch}, which currently focus on static datasets. 

ANN quantization methods have also been designed for specific measures such as inner products \cite{guo2020accelerating} and bitwise operations \cite{gao2024RaBitQ,aguerrebere2024locally}. RaBitQ \cite{gao2024RaBitQ} quantizes points using randomly rotated bi-valued unit vectors, and is extended \cite{gao2024practical} to support more bits per dimension. 
That line of work targets larger sizes, typically $\geq1$ bit per dimension, while PQ and our method \ouralg~target a smaller size regime, typically $<1$ bit per dimension.

Quantization is commonly used to complement index structures like HNSW \cite{malkov2018efficient} and IVF \cite{jegou2010product} by enabling compressed representations in memory. Hybrid systems \cite{matsui2018survey, chen2021spann, pan2024survey, jayaram2019diskann, gou2025symphonyqg} leverage this combination for memory efficiency and avoiding I/O during queries. \ouralg~can serve as a drop-in replacement for existing quantization methods in such architectures, e.g., replacing PQ in DiskANN \cite{jayaram2019diskann} or RaBitQ in SymphonyQG \cite{gou2025symphonyqg}, and provide guarantees of dynamic consistency and update performance under streaming updates.

\vspace{-3pt}
\paragraph{Incremental clustering for ANN search.}
Recent incremental clustering methods for ANN search include Online PQ \cite{online-pq}, DeDrift \cite{dedrift}, and SPFresh \cite{spfresh}.\footnote{Of these, only Online PQ targets incremental clustering for the purpose of quantization, while DeDrift and SPFresh are designed for the IVF component of the ANN pipeline~\cite{jegou2010product}; however, they can be adapted for dynamic quantization as we implemented for our experimental evaluation.}  
Online PQ \cite{online-pq} updates cluster centers by assigning new points to the closest cluster and updating the corresponding center accordingly. However, it never reassigns cluster membership, reducing the effectiveness and adaptability of the quantization over time. 
DeDrift \cite{dedrift} updates the cluster centers as new points arrive and aims to handle the cluster imbalance by moving points from large clusters to newly created clusters. However, DeDrift does not address disk storage and the interactions between disk and main memory, which is critical for large datasets.  SPFresh \cite{spfresh} uses approximately-balanced clustering and applies a heuristic to identify vectors for cluster reassignment after updates. This heuristic may still require retrieving a large proportion of vectors, or even the entire dataset, from disk. \ouralg~addresses these limitations by introducing a dynamically consistent mechanism for quantization that minimizes the number of vectors switching partitions during updates and enables efficient synchronization between in-memory sketches and disk-resident data. 

\vspace{-3pt}
\paragraph{I/O and update complexity of ANN search.}
Goswami et al.~\cite{goswami2020complexity} study the I/O complexity of exact and approximate k-NN search 
in the external memory \cite{miltersen1999cell} and indexability \cite{hellerstein1997analysis} models, focusing on the utility of block transfers. 
They prove that polynomial space indexing schemes for exact k-NN is unable to utilize block transfers, requiring $\Omega(k)$ block reads regardless of the block size $B$, and that this bound holds in $\ell_\infty$ spaces even for approximate $k$-ANN with approximation factor $c < 3$.
Conversely, for $c\geq3$, they show that indexing with $\lceil k/B \rceil$ I/Os per query is possible for every metric in the indexability model. 
As they state, their model captures coordination between servers in a distributed search system where each query requires multiple I/Os, whereas our model is meant to capture search on each server with zero or one I/O per query. 
Kapralov~\cite{kapralov2015smooth} establishes smooth insertion-query tradeoffs for approximate nearest neighbor search with locality-sensitive hashing (LSH), with insert complexity $O(N^{\alpha^2\rho_\alpha})$ and query complexity $O(N^{(1-\alpha)^2\rho_\alpha})$, where $\rho_\alpha = \frac{c^2+(1-\alpha)^2}{4-3\alpha^2}$. This is meant to balance computational costs between query and insertion, in terms of running time, though not to save memory by quantization, as LSH indexes generally entail super-linear storage costs in memory. 
In a similar vein, Yi~\cite{yi2009dynamic}  studies dynamic range indexing, which can be seen as low-dimensional search, and proves that query cost $O(q + K/B)$ and insertion cost $O(u/B)$ must satisfy $q \cdot \log(u/q) = \Omega(\log B)$ for $q < \alpha \ln B$. 
Also relatedly, Andoni et al.~\cite{andoni2017optimal} and Christiani \cite{christiani2017framework} establish optimal time-space trade-offs for ANN, though they do not touch on dynamic updates, and again use super-linear space.

\bibliographystyle{plain}
\bibliography{refs}


\clearpage
\appendix
\section{Data-Dependent vs.~Data Oblivious ANN Quantization}\label{app:datadependent}

For background, we give a brief overview of the notion of data-dependence in ANN quantization. Let $X\subset\mathbb R^d$ be the search dataset. An ANN quantizer is a (possibly randomized) data structure $D$ that assigns to every point $x\in X$ a quantized proxy $\tilde x$, such that the proxies admit an efficient bit representation and can be fully stored in main memory. If the quantized proxy $\tilde x$ is a function of only the point it describes, i.e., $\tilde x=f(x)$, then $D$ is a \emph{data-oblivious} quantizer. 
If $\tilde x$ also depends on the rest of the dataset, i.e., $\tilde x=F(x,X)$, then $D$ is a \emph{data-dependent} quantizer.

\paragraph{Examples.}
\begin{CompactItemize}
    \item Scalar Quantization (SQ) is an example of a deterministic data-oblivious quantizer. In SQ, for every $x\in X$, the quantized proxy $\tilde x$ is generated by rounding each coordinate of $x$ to a pre-specified bit precision. This does not involve other points in $X$.
    \item Classical locality-sensitive hashing (LSH) \cite{lsh} is an example of a randomized data-oblivious quantizer.\footnote{The original and primary use of LSH in ANN is not quantization, but for pruning the search dataset at query time. However, it can also be used formally for quantization and is often used in this way.} For example, in E2LSH \cite{e2lsh}, the points are projected onto a random direction, shifted by a random shift, and then SQ is performed on the projected and shifted coordinates. Even though the projection and the shift are chosen at random, they are chosen from fixed distributions which are independent of the dataset $X$, and thus each quantized proxy $\tilde x$ is a (randomized) function only of $x$. 
    \item $k$-means is an example of a data-dependent quantization (as well as many quantizers based on it, like PQ). The quantized proxy $\tilde x$ is the centroid of the cluster that contains $X$. Since the partition into clusters and the cluster centroids depend on the entire dataset $X$, this is a data-dependent quantizer. 
    \item \ouralg\ is also a data-dependent quantizer. In \ouralg, the leaves of the kd-tree induce a clustering of $X$, and the proxy $\tilde x$ of each point $x$ is the mean of the leaf-cluster it belongs to. Since the kd-tree partition is determined by coordinate medians, which depend on the entire dataset $X$, this is a data-dependent quantizer.  
\end{CompactItemize}

\paragraph{Handling updates.}
For data-oblivious quantizers, point insertion and deletion are trivial, since removing $\tilde x_i=f(x_i)$ from $X$ does not impact $\tilde x_j=f(x_j)$ for $i\neq j$. In the terminology of \Cref{sec:model}, they are trivially \emph{dynamically consistent} without disk access. 
For data-dependent quantizers, updates create a difficulty. If a point $x_i$ is inserted, then formally, for every other $x_j\in X$, the proxy $\tilde x_j$ needs to be updated from $F(x_j, X)$ to $F(x_j, X\cup\{x_i\})$. Similarly, if $x_i$ is deleted, $\tilde x_j$ needs to be updated from $F(x_j, X)$ to $F(x_j, X\setminus\{x_i\})$. Dynamic consistency captures this requirement formally.

Dynamic consistency can be achieved trivially by building the quantizer from scratch, but such a rebuild is computationally expensive and requires reloading the full precision of all of $X$ from the remote disk. 
The meaningful challenge is to achieve dynamic consistency with limited disk access per update. This is the challenge our DySk model is aimed at formulating and our algorithms in this paper are aimed at addressing.

\section{Provably Accurate DySk-Based ANN Quantization with Dynamic Updates}\label{app:quadsketch}
In this appendix we include the omitted details from \cref{sec:quadsketch}. We formally define the $(\eps, \delta)$-dynamic approximate nearest neighbor search problem ($(\eps, \delta)$-DANN) and we describe the provably accurate and I/O efficient DySk-based dynamic data structure that solves $(\eps, \delta)$-DANN. 
\subsection{Dynamic approximate nearest neighbor search}
The \emph{static} nearest neighbor data structure problem for a the point set (dataset) $X \subset U = \mathbb{R}^d$ requires us to design a data structure $D$ that returns the nearest neighbor in $X$ to a query point $y \in \mc{F} := \mathbb{R}^d$.
Recall that such a data structure has two goals: 1) sketch (or compress) the represent of the points in $X$, and 2) return a nearest neighbor candidate that is ``accurate''.
The particular notion of accuracy that we will be interested in theoretically is the following.
\begin{definition}[$(\eps,\delta)$-ANN]
Fix $\eps, \delta \in (0,1)$.
We say that a randomized data structure $D$ solves the \emph{$(\eps,\delta)$-approximate nearest neighbor search}
\footnote{We will only consider the Euclidean distance.} 
problem if for any $X \subset \mathbb{R}^d$ and any query point $y \in \mathbb{R}^d$ we have
\[
\Pr[\|\hat x - y\|_2 \le (1+\eps)\min_{x^\star \in X}\|x^\star - y\|_2 ] \ge 1-\delta,
\]
where $\hat x \in X$ is the nearest produce by $D$ on query $y$.
Here, the probability is taken over the randomness of the data structure $D$.
\end{definition}
A sequence of work by \cite{indyk2017near,indyk2018approximate,indyk2022optimal} produces efficient algorithms that achieving strong compression guarantees for this problem.
Unfortunately, these solutions do not immediately solve the following natural extension of the $(\eps,\delta)$-ANN problem to the dynamic setting.
\begin{definition}[$(\eps,\delta)$-DANN]
Fix $\eps, \delta \in (0,1)$.
We say that a randomized dynamic data structure $D$ solves the \emph{$(\eps,\delta)$-dynamic approximate nearest neighbor search} problem if it achieves the following:
\begin{enumerate}
    \item $D$ can support an arbitrary sequence of insertion into and deletions from $X$.
    \item $D$ satisfies the $(\eps,\delta)$-ANN at any point in time. 
\end{enumerate}
\end{definition}
The main result of \cref{sec:quadsketch} (\cref{thm:quadsketch}) stated that there is a efficient and dynamically consistent data structure that solves the $(\eps,\delta)$-DANN problem. \cref{thm:quadsketch} is a special case of \cref{thm:quadsketch_general} below.
The data structure whose guarantees are stated in \cref{thm:quadsketch_general} is a modification of \emph{QuadSketch}.  
Specifically, while \cite{quadsketch} only prove in-sample guarantees for their data structure (i.e., the query $q$ must be one of the data points in $X$), it can be combined with the out-of-sample extension for quantized $(\epsilon,\delta)$-ANN from \cite{indyk2018approximate} to remove this limitation (see \cite{wagner2020efficient}, Chapter 4 for details).
The QuadSketch data structure is based on \emph{Quadtrees}.
We briefly outline the details of QuadSketch and will focus on emphasizing the aspects of QuadSketch that will be relevant for our dynamically consistent version of it in the DySk model. 
Let
\[
\Phi = \Phi(\mc{U}) = \frac{\max_{x,y \in \mc{U}} \|x-y\|_2}{\min_{\substack{x , y \in \mc{U}\\ x\not=y}}\|x-y\|_2}
\]
denote the aspect ratio of the universe $\mc{U}$ which is the ratio between the largest and smallest distance of the points in $\mc{U}$. 
We will make the natural assumptions that every vector $x$ in $\mc{U}$ satisfies $\|x\|_{\infty} \le \phi$ and that aspect ratio satisfies $\Phi \ge \phi$.

The QuadSketch data structure compresses $X$ as follows:
\begin{enumerate}
    \item Set $\Delta = 2^{\lceil \log \Phi\rceil}$. 
    Let $H'$ be a $d$-dimensional grid with side length $4\Delta$ centered at the zero vector $0 \in \mathbb{R}^d$. 
    Notice this grid contains the entire data set $X$ by our assumption on $\mc{U}$.
    \item Let $\sigma_1, \dots, \sigma_d$ be independent uniform random points from $[-\Delta, \Delta]$ and let $H$ be a random shift of $H'$ where we shift the $i$th coordinate by $\sigma_i$.
    For every integer $ -\infty < \ell \le \log_2(4 \Delta)$ we write $G_\ell$ to denote the axis-parallel grid with cell side $2^\ell$ which is aligned with $H$.
    Notice that this shifted grid still contains all points in $X$.
    \item A $2^d$-ary quadtree is built on the (nested) grids $G_\ell$ as follows: for every cell $C \in G_\ell$ with a tree node at level $\ell$ such that its children are the $2^d$ grid cells in $G_{\ell-1}$ that are inside the cell $C$. 
    \item The edge connecting $v$ to its child $v'$ is labelled with the string of $d$ bits as follows: the $j$th bit is $0$ if $v'$ coincides with the bottom half of $v$ along coordinate $j$ and $1$ otherwise.
    \item The data structure defines such a quadtree by associating the root with $H$ and recursively building the tree.
    Crucially, when we are constructing the tree we only include children nodes if their corresponding cell contains at least 1 point and we stop building the tree at a (to be defined) level $L$.
    The final tree is denoted $T$.
    \item Let $\Lambda = \log(\eps^{-1}\delta^{-1}16d^{1.5}\log \Phi)$.
    We say that a downward path $P = (v_0, v_1 , \dots, v_l)$ in $T$ satisfying $\deg(v_1) = \dots =  \deg(v_{k-1}) = 1$ and $\deg(v_0) , \deg(v_k) \ge 2$ is \emph{maximal} if $k \ge 1$, where $k$ is the length of the maximal path. Set the depth of the tree to $L = \log \Phi + \Lambda$.
    \item We apply \emph{middle-out compression} to every maximal path $P$ in $T$ of length $k \ge 2\Lambda + 1$: we remove the nodes $v_{\Lambda+1} \dots, v_{k-\Lambda}$ with all their adjacent edged (and labels) and connect $v_\Lambda$ directly to $v_{k-\Lambda+1}$.
    We call such an edge \emph{long}.
    Moreover, we label every long edge with its edge length (based on how many points were removed).
    \item Denote by $T'$ the final compressed tree defined by the previous steps. A simple argument (see \cite{quadsketch}) shows that this sketch can be stored exactly in memory using $O(nd\Lambda + n\log n)$ bits of space.
    \end{enumerate}
The dataset $X$ is a subset of the universe $\mc{U}$ with $|\mc{U}| \le N$, so a standard application of the Johnson-Lindenstrauss lemma tells us that we can assume that $d = O(\frac{\log N}{\eps^2})$ \cite{nelson2020sketching}, Chapter 5.
The following result shows how we can adapt QuadSketch to our setting.
\begin{theorem}\label{thm:quadsketch_general}
There is a dynamically consistent version of QuadSketch that solves the $(\eps,\delta)$-DANN problem in the DySk model with the following guarantees:
    \begin{itemize}
        \item \textbf{Main memory space:} $O(n\frac{\log N}{\eps^2}(\log(\frac{1
        }{\eps\delta}) + \log\log N +\log\Phi)  + n\log n)$ bits.
        \item \textbf{I/Os per-update:} One read and one write I/O each of size $O(\eps^{-2}\log N)$ words.
        \item 
        \textbf{Disk Space:} $O(n \eps^{-2}\log N)$ words.
        \item \textbf{Update and Query Time:} $O\left(\frac{\log N}{\eps^2}\left( \log(\frac{1}{\eps\delta}) + \log\log N + \log \Phi \right)\right)$.
    \end{itemize}
\end{theorem}
\begin{proof}
We begin by introducing some more notation.
We write $T(X)$ to specify that the quadtree was built on the dataset $X$.
for a maximal path $(v_0, v_1, \dots, v_{\ell-1}, v_\ell)$ we will sometimes refer to $v_0$ as the head and $v_\ell$ as the tail.

\textbf{Build:} We build and store the compressed quadtree $T'(X)$ in main memory exactly as was described in the description of QuadSketch.
Furthermore, for each $x_i \in X$ we write to a unique disk address $\addr(i)$ the full precision of $x_i$.
We only store the QuadSketch and the disk addresses $\{\addr(i)\}_{i \in [n]}$ in main memory.

\textbf{Main memory Space:} From Chapter 4 of~\cite{wagner2020efficient} we know that we only need to use 
\[
O\left(n\frac{\log N}{\eps^2}\left(\log\left(\frac{1}{\eps\delta}\right) + \log\log N+\log\Phi\right)  + n\log n\right)\]
bits of space to store the tree $T'(X)$ in main memory.
We also use $n\log N$ bits to store the disk addresses $\{\addr(i)\}$ for each $x_i$ which we don't count towards our memory usage.)

\textbf{Disk Space:} We only store the full precision vectors on disk which requires $O(nd) = O\left(n\frac{\log N}{\eps^2}\right)$ words of memory. (Recall that we assume that the number of bits $\bits$ required to represent each coordinate can fit in a word).

\textbf{Insertion:}
Let $x \in \mc{U}$ be the new point that is to be inserted into $X$.
Let $v_x$ be the leaf corresponding to $x$ in $T(X \cup \{x\})$.
Our goal will be to insert $x$ into the pruned tree $T'(X)$ to directly obtain the pruned tree $T'(X\cup\{x\})$ without having to build the entire tree $T(X \cup \{x\})$.
In particular, we will show that we can recover some of the discarded labels on long edges on $T'(X)$ using a single round of I/O.

Define $v_x$'s  \emph{branching node} $v^*$ in $T(X \cup \{x\})$ to be $v_x$'s deepest ancestor that has degree larger than 1.
Our first step towards this goal will be to determine the branching node $v^*$ as well as all bits on the edges of the \emph{prefix path} $P_< = (v_r, \dots, v^*)$.
We will also find it useful to define the \emph{suffix path} $P_> = (v^*, \dots, v_x)$.
We begin by traversing the tree $T'(X)$ starting from the root $v_r$ along the downward path through $v_r$'s child that corresponds to the grid cell corresponding to $x$ (the one in which $x$ lands in on this level).
We continue this process going down the tree. 
There are two possibilities that differ from the above when traversing the tree. The first is that we reach a node $v'$ for which all its outgoing edges are labeled with bits that disagree with $x$.
In this case we pick an arbitrary child and continue the downward traversal.
The second one is that we run into a \emph{long edge} that has been compressed.
In this case we pretend that the bits that were removed correspond exactly to the bits that the point $x$ would have assigned to those edges.
We perform this traversal until we get to some leaf node corresponding to a point $x_i \in X$.
We then perform a read I/O to obtain the full precision of $x_i$ from disk.

Our first claim is that we have recovered the identity of $v^*$ as well as all the bits on the edges of the prefix path $P_< = (v_r, \dots, v^*)$.
If we determine $v^*$ before reaching a long edge the claim is immediate, so we will assume that this does not happen.
Consider the case that we have incorrectly filled-in the bits corresponding to at least one long edge we have come across.
In this case, we can use $x_i$ to determine the correct bits of any long edged for which we had incorrectly filled-in. 
In particular, we have that $v^*$ is a node that was compressed on the \emph{first} incorrectly filled-in long edge.
It is also easy to see that the claim holds when we fill in all long edges correctly.
Now that we have determined the prefix path $P_<$ and its associated edge labels, we insert the suffix path $P_>$ and its corresponding edge labels into $T'(X)$ and fill in all the bits deleted from long edges on the root-to-leaf path corresponding to $x_i$. 
From this we obtain the (partially compressed) tree $T_{\text{par}}$ and we apply middle out compression to it.

We will now argue that compressing $T_{\text{par}}$ yields the tree $T'(X \cup \{x\})$.
To argue this, we will consider two cases for the degree of the branching vertex $v^*$. Case 1): $\deg(v^*) > 1$ in $T(X)$. We have that all long edges not on $P_>$ remain the same by definition and that the only new long edges are on $P_>$. 
From this observation it is immediate that compressing $T_{\text{par}}$ yields $T'(X \cup \{x\})$. Case 2): $\deg(v^*) = 1$ in $T(X)$. We see that the insertion of $x$ has caused the degree of $v^*$ to increase.
As a consequence, this has altered \emph{a single pre-existing} maximal path $P = (v_h, \dots, v^*, \dots, v_t)$ of the original tree $T(X)$.
Notice that the path $P$ is part of the root-to-leaf path of $x_i$ by definition of our traversal procedure and the fact that $\deg(v^*) = 1$ in $T(X)$.
Thus, we have recovered the labels of all the edges on $P$ who's compression might change from the addition of $x$.
It is now straightforward to see that compressing $T_{\text{par}}$ yields $T'(X \cup \{x\})$.

Finally, we can write to a new disk address $\addr(n+1)$ the full precision of $x$.
Since the above procedure recovers $T'(X\cup \{x\})$, we conclude that the the delete procedure is dynamically consistent.The running time of the insert procedure is $O\left(\frac{L\log N}{\eps^2}\right) = O\left(\frac{\log N}{\eps^2}\left( \log(\frac{1}{\eps\delta}) + \log\log N + \log \Phi \right)\right)$.

\textbf{Delete:} 
Let $x_i \in X$ be the point that is to be deleted and let $v_{x_i}$ be its corresponding leaf node in the tree $T(X)$.
Our goal will be to directly delete $x_i$ from the pruned tree $T'(X)$ in order to obtain the pruned tree $T'(X\setminus\{x_i\})$ without having to build the entire tree $T(X \setminus \{x_i\})$.
In particular, we will show that we can recover some of the discarded labels on long edges on $T'(X)$ using a single round of I/O.

Similar to the insert procedure, our first step will be to determine the the branching node $v^*$ that is $v_x$'s deepest ancestor in $T(X)$ with degree larger than 1.
Notice that we can immediate determine this by traversing the compressed tree $T'(X)$ starting from the root $v_r$ through the child corresponding to  $x_i$'s grid cell since we have the full precision of $x_i$.
Once we have determined $v^*$, we continue our traversal from $v^*$ through any of its children that \emph{does not} correspond to $x_i$'s grid cell all the way down to an arbitrary leaf corresponding to some $x_j \in X$.
We then perform a read I/O to get  $x_j$ from $\addr(j)$.
Finally, we remove the branch of $T'(X)$ starting from $v^*$'s child that corresponds to $x_i$.
Put differently, we remove $x_i$'s contribution from the tree.
From this we get the (partially compressed) tree $T_{\text{par}}$ and apply middle out compression to it.

We will now argue that compressing $T_{\text{par}}$ yields the tree $T'(X\setminus \{x_i\})$.
We consider two cases for the degree of the branching vertex $v^*$. Case 1): $\deg(v^*) > 2$ in $T(X)$.
In this case no long edges that remain in the tree have changed from deleting $x_i$ thus compressing $T_{\text{par}}$ yields $T'(X \setminus \{x_i\})$. Case 2): $\deg(v^*) = 2$ in $T(X)$.
In this case removing $x_i$ implies $\deg(v^*) = 1$ in $T(X \setminus \{x_i\})$.
Thus, $v^*$ could be 1) some internal node of a maximal path \emph{or} 2) the head of some maximal path and/or the tail of some maximal path. Moreover, these are the only potentially maximal paths that have changed.
Regardless of which (sub)case we are in, we have the labels of all relevant edges from the full precision of $x_j$ whose root-to-leaf path goes through these maximal paths.
Thus we can see that compressing $T_{\text{par}}$ yields $T'(X \setminus \{x_i\})$.

Finally, we can request a write probe to delete the contents in $\addr(i)$ that contains the full precision of $x_i$.
Since the above procedure recovers $T'(X\setminus \{x_i\})$, we conclude that the the delete procedure is dynamically consistent.
Furthermore the running time of the delete operation is $O\left(\frac{L\log N}{\eps^2}\right) = O\left(\frac{\log N}{\eps^2}\left( \log(\frac{1}{\eps\delta}) + \log\log N + \log \Phi \right)\right)$.

\textbf{Query:} The data structure maintains a copy of the QuadSketch in memory so we can run the query algorithm defined in \cite{wagner2020efficient}, Chapter 4 which runs in $O\left(\frac{L\log N}{\eps^2}\right) = O\left(\frac{\log N}{\eps^2}\left( \log(\frac{1}{\eps\delta}) + \log\log N + \log \Phi \right)\right)$ time.
\end{proof}

\section{Our \ouralg\ Method}\label{app:tbq}
In this appendix we include all the omitted proofs from \cref{sec:tbq}.
We also include some background information on PQ and explain how to combine use \ouralg~with it.
\subsection{Stability of kd-trees}\label{app:tbq_stability}
In this section we prove that the partition induced by a kd-tree satisfy a certain notion of \emph{stability}: whenever a single point is inserted into/deleted from the tree, every internal node changes by at adding a single new point and/or deleting a single existing point.
This result is central to the proof of~\cref{thm:tbq}.
We begin with some notation.

Let $X = \{x_1, \dots, x_n \} \subseteq \mc{X}$ be a set of $n$ elements that live in a universe $\mc{X}$ that admits a strict total ordering. 
Define the median index $I(n) = \lceil\frac{n}{2}\rceil$ and consider a partition of $X$ based on the median where the parts of the partition are defined as
\begin{gather*}
L(X) = \{x_{(1)}, \dots, x_{(I(n)-1)}\} \quad \\
\text{and} \quad R(X) = \{x_{(I(n))}, \dots, x_{(n)}\},
\end{gather*}
where the notation $x_{(i)}$ denotes the $i$th smallest point in $X$.
The following proposition highlights the stability of a median based partition when one point is added and/or deleted from the dataset.
\begin{proposition}\label{prop:stability_median}
Fix an integer $n \ge 1$ and a set $X = \{x_1, \dots, x_n\} \subset \mc{X}$ where the universe $\mc{X}$ admits a strict total ordering.
Let $X^+ \subset \mc{X}$ be a set containing at most one element such that any element in $X^+$ is not in $X$.
Let $X^- \subset \mc{X}$ be a set containing at most one element such that any element in $X^-$ is in $X$.
The set $X' = \left(X \setminus X^{-} \right) \bigcup X^{+}$ satisfies
\begin{gather}\label{eq:setminus_stability1}
    |L(X) \setminus L(X')| \le 1 \quad
    \text{and} \quad |R(X) \setminus R(X')| \le 1,
\end{gather}
and
\begin{gather}\label{eq:setminus_stability2}
    |L(X') \setminus L(X)| \le 1 \quad
    \text{and} \quad |R(X') \setminus R(X)| \le 1.
\end{gather}
\end{proposition}
\begin{proof}
We consider 4 cases, each of which may have some further subcases.
Let $x'_{(1)} \le x'_{(2)} \le \dots \le x'_{(n+1)}$ be the elements of $X'$ in sorted order.

\textbf{Case 1:} Assume $|X^{+}| = |X^{-}| = 0$. 
Then $X' = X$ and the claim trivially holds.

\textbf{Case 2:} Assume $|X^{+} | = 1$ and $|X^{-}| = 0$, i.e.\ we will insert a single point into the set $X$.
Let $x^+$ be the point in $X^+$.
Let $I' = I(n+1)$ and $I = I(n)$.
Either $I' = I+1$ or $I' = I$.
We four exhaustive subcases.

\textbf{Case 2.i:} Assume that $I' = I+1$ and that $x^+ < x_{(I)}$. We see that $x'_{(I')} = x'_{(I+1)} = x_{(I)}$
and thus 
\begin{gather*}
L(X') = \{x^+\} \cup \{x_{(1)}, \dots, x_{(I-1)}\} \quad \\
\text{and} \quad R(X') = \{x_{(I)}, \dots, x_{(n)}\}.
\end{gather*}

\textbf{Case 2.ii:} Assume that $I' = I+1$ and that $x^+ > x_{(I)}$. We see that $x'_{(I'-1)} = x'_{(I)} = x_{(I)}$ and thus 
\begin{gather*}\label{eq:add_annoying_case_left}
L(X') = \{x_{(I)} \} \cup \{x_{(1)}, \dots, x_{(I-1)}\} \quad \\
\text{and} \quad  R(X') = \{x^+\} \cup \{x_{(I+1)}, \dots, x_{(n)}\} \setminus \{ x_{(I)}\}.
\end{gather*} 

\textbf{Case 2.iii:} Assume that $I' = I$ and that $x^+ < x_{(I-1)}$. We see that $x'_{(I')} = x'_{(I)} = x_{(I-1)}$, 
and thus 
\begin{gather*}
L(X') = \{x^+\} \cup \{x_{(1)}, \dots, x_{(I-1)}\} \setminus \{x_{(I-1)} \} \quad \\
\text{and} \quad R(X') =  \{x_{(I-1)} \}\cup \{  x_{(I)}, \dots, x_{(n)}\}
\end{gather*}
\textbf{Case 2.iv:} Assume that $I' = I$ and that $ x^+ > x_{(I-1)}$. We see that $x'_{(I'-1)} = x'_{(I-1)} = x_{(I-1)}$ and thus
\begin{gather*}
L(X') = \{x_{(1)}, \dots, x_{(I-1)}\} \quad \\
\text{and} \quad R(X') = \{x^+\} \cup \{x_{(I)}, \dots, x_{(n)}\},
\end{gather*}
In all of the above subcases of \textbf{Case 2} we can conclude that \cref{eq:setminus_stability1,eq:setminus_stability2} both hold.

\textbf{Case 3:} Assume $|X^{+}| = 0$ and $|X^{-}| = 1$, i.e.\ we will delete a single point from the set $X$.
The claim follows from applying the conclusion of \textbf{Case 2} with the ``initial'' set $\widetilde{X} = X \setminus X^{-}$, the ``insertion'' set $\widetilde{X}^+ = X^-$, and the ``resulting'' set $\widetilde{X}' = \widetilde{X} \cup \widetilde{X}^+$.

\textbf{Case 4:} Assume $|X^{+}| = |X^{-}| = 1$, i.e.\ we insert a single point \emph{and} delete a single point.
Let $x^+$ be the point in $X^+$ and let $x^-$ be the point in $X^-$.
It is easy to see that $|X'| = |X|$ so the median index $I = I(n)$ remains the same.
We consider four exhaustive subcases.

\textbf{Case 4.i:} Assume that $x^+ < x_{(I)}$ and $x^- < x_{(I)}$.
We have that $x'_{(I)} = x_{(I)}$ and thus 
\begin{gather*}
L(X') = \{x^+\} \cup \{x_{(1)}, \dots, x_{(I-1)}\} \setminus \{x^-\} \quad \\
\text{and} \quad R(X') = \{x_{(I)}, \dots, x_{(n)}\}.
\end{gather*}

\textbf{Case 4.ii:} Assume that $x_{(I-1)} < x^+$ and $ x_{(I-1)} < x^-$.
We have that $x'_{(I-1)} = x_{(I-1)}$ and thus 
\begin{gather*}
L(X') =  \{x_{(1)}, \dots, x_{(I-1)}\} \quad \\
\text{and} \quad R(X') = \{x^+\} \cup \{x_{(I)}, \dots, x_{(n)}\} \setminus \{x^-\}.
\end{gather*}
\textbf{Case 4.iii:} Assume that $x^+ < x_{(I-1)} < x_{(I)} \le x^- $.
We have that $x'_{(I)} = x_{(I-1)}$ and thus 
\begin{gather*}
L(X') =  \{x^+\} \cup \{x_{(1)}, \dots, x_{(I-1)}\} \setminus \{x_{(I-1)}\} \quad \\
\text{and} \quad R(X') = \{ x_{(I-1)}\} \cup \{ x_{(I)}, \dots, x_{(n)}\} \setminus \{x^-\}.
\end{gather*}
\textbf{Case 4.iv:} Assume that $x^- \le x_{(I-1)} < x_{(I)} < x^+ $.
We have that $x'_{(I-1)} = x_{(I)}$ and thus 
\begin{gather*}
L(X') =  \{ x_{(I)} \} \cup \{x_{(1)}, \dots, x_{(I-1)}\} \setminus \{x^-\} \quad \\
\text{and} \quad R(X') = \{x^+ \} \cup \{x_{(I)}, x_{(I)}, \dots, x_{(n)}\} \setminus \{ x_{(I)} \}.
\end{gather*}
In all of the above subcases of \textbf{Case 4} we can conclude that \cref{eq:setminus_stability1,eq:setminus_stability2} both hold.
\end{proof}
An immediate consequence of the above is that if we insert or delete a single point from the root of a kd-tree, we can bound the number of changes in every node.
The proof follows from an simple inductive argument that utilizes Proposition \ref{prop:stability_median}.
The following is \cref{thm:kd-tree-stability} restated.
\begin{theorem}[kd-tree stability]
Fix a dataset $X \subset \mathbb{R}^d$ and integer $L > 0$. 
Let $T$ be a depth $L$ kd-tree built on $X$. 
Let $X'$ be the dataset after inserting or a deleting a point.
In either case, every node of the new kd-tree $T'$ has changed by removing at most one point from and/or inserting at most one point into the corresponding node in $T$.
Furthermore, the only points in both $X$ and $X'$ that can change their root-to-leaf path after the insert/delete are the maximum (resp. minimum) element of a node $v$ which is a left (resp. right) child in $T$.
\end{theorem}
\begin{proof}
Let $X_v \subseteq X$ be the subset of $X$ that lands in the node $v$ of $T$.
We define $X'_v \subset X'$ the  same way.
The first claim of the theorem is equivalent to the claim that after an insert/delete, for every depth $\ell$ and every node $v$ in $T$ at depth $\ell$, we have 
\[ \max\{|L(X_v) \setminus L(X'_v)|, |L(X'_v) \setminus L(X_v)| \}\le 1 \]
\begin{equation}\label{eq:stability_induction}
\text{and} \quad \max\{|R(X_v) \setminus R(X'_v)|, |R(X'_v) \setminus R(X_v)| \}\le 1,
\end{equation}
where the points in $X_v$ are ordered with respect to their $j_\ell$-th coordinate.
The second claim of the theorem is equivalent to the following: after an insert/delete, for every depth $\ell$ and every node $v$ in $T$ at depth $\ell$,  the only points in $X_v$ that move from $L(X_v)$ to $R(X_v)$ or vice versa are $x_{(I(|X_v|)-1)}$ or $x_{I(|X_v|)}$. 
Formulating the claims this way will be more convenient.

We will prove both claims via induction over the depth $\ell \in \{1, \dots, L\}$ of the tree.
For the base case $\ell=1$, it is easy to see that both cases hold.
For the inductive step, assume that the claim is true for nodes at level $\ell-1$.
We will show it is true for nodes at level $\ell$.
Fix a node $v_P$ at level $\ell-1$ of the tree.
Let $v_L$ be the left child of $v_P$.
By the inductive assumption we have $\max\{|L(X_{v_P}) \setminus L(X'_{v_P})|, |L(X'_{v_P}) \setminus L(X_{v_P})| \}\le 1$.
Since $X_{v_L} = L(X_{v_P})$ by definition, we conclude that at most one point is inserted into and/or at most one point is deleted from $v_L$. \cref{prop:stability_median} now implies that \cref{eq:stability_induction} holds for $v_L$.
Furthermore, inspecting the detailed case analysis of \cref{prop:stability_median} reveals that the only points that move from $L(X_{v_L})$ to $R(X_{v_L})$ or vice versa are $x_{(I(|X_{v_L}|)-1)} \in X_{v_L}$ or $x_{(I(|X_{v_L}|))} \in X_{v_L}$.
An identical argument yields the same conclusion for $v_P$'s right child $v_R$. 
As the choice of $v_P$ was arbitrary, we conclude that both claims hold for all nodes at depth $\ell$.
\end{proof}

\subsection{DySk-based Priority Queues}\label{app:tbq_priority}
In this section we will prove that there is an efficient implementation of a max priority queue based on max heaps in the DySk model.
Before we do this, it will be helpful to recall what max heaps are as well as some of their properties.

We say a tree is built on a dataset $X = \{x_1, \dots, x_n\}$ if the nodes in the tree are in one-to-one correspondence to a data element in $X$.
In other words, they key of a node corresponds to a unique data element $x_i \in X$.
A \emph{max heap} built on dataset $X$ is a complete binary tree that satisfies the \emph{heap property}: every internal node in the tree has a larger value than both of its children.
Recall how a point is inserted into/deleted from a max heap. 

\textbf{Insertion:} When inserting $x$, we go to the first vacant leaf and insert $x$ there. If $x$ is larger than its parent $x'$, we swap $x$ and $x'$ in the tree. We continue this process on the upward path to the root until we cannot swap anymore or $x$ becomes the new root.

\textbf{Deletion: } To delete $x \in X$ from the tree, we swap $x$ with the last leaf in the tree with corresponding key $\tilde{x} \in X$.
We then delete $x$ from the tree (i.e.\ from the last leaf) and consider 2 cases: $1)$ if $\tilde{x}$ is larger than its parent $x'$, swap it or $2)$ if $x$ is smaller than its larger child $x'$, swap them. 
We continue this process of moving $x$ up or down the tree until neither condition $1)$ nor $2)$ is satisfied.

We are now ready to prove the \cref{lem:max_DDS}.
\begin{proof}[Proof of \cref{lem:max_DDS}]
We will maintain a max priority queue by implementing a max heap in the DySk model.
Before we explain how to do this, we introduce some definitions.
Let $P$ be the unique path from the root to some node $v$ in a binary tree.
We define the \emph{off-by-one tree of $P_v$}, denoted $T(P)$, to be the subtree that contains all the siblings of all nodes on the path $P$.
It is not hard to see that this subtree has size at most $2\ell -1$ where $\ell$ is the length of $P$.
We define the \emph{insertion path} $P^*$ of a max heap to be the path starting from root to the parent of the first vacant leaf in the heap.
For any node $x$ in a max heap, we define the \emph{deletion path} of $x$, denoted $P_x$, to be the unique root-to-leaf path that goes from the root to $x$, and then proceeds downward through the larger child until it reaches a leaf.
We should emphasize that we do not view these paths as a sequence of values from the dataset, but rather a sequence of nodes in the actual tree.

\textbf{Build:} Given an initial dataset $X = \{x_1, \dots, x_n\}$, we build the max heap as follows.
First, for every element $x_i \in X$ we assign a unique disk address $\texttt{addr}(i)$ requiring at most $\log N$ bits of space whose contents will be determined below.
We will assume that every disk address points to a block of memory that consists of $\Theta(\log N)$ words.\footnote{This assumption is harmless since we can just make a convention to reserve a large enough contiguous blocks of disk memory for every disk address.}
After this, we construct a max heap on $X$ where we slightly modify the keys of each heap node $v$ corresponding to some $x_i$ to be the tuple $(x_i, \texttt{addr}(i))$.
Crucially, comparisons are only performed using the $x_i$ part of the tuple.
For every node $v$ we assign an auxiliary disk address $\texttt{a}_v$.
For every non-root node $v$ we write to $\heapaddr_v$ the key $(x_i, \addr(i))$ and the sequence of disk addresses $(\heapaddr_u)_{u \in P_v}$ of the nodes on the deletion path $P_v$. 
For the root node $v_r$, we write to the disk address $\texttt{a}_{v_r}$ the tuple $(x_{i^*}, \texttt{addr}(i^*))$ (recall that $x_{i^*}$ denotes the maximum), two sequences of (heap) disk addresses $(\heapaddr_u)_{u \in P_{v_r}}$ and $(\heapaddr_u)_{u \in P^*}$ (corresponding to the  deletion and insertion paths $P_{v_r}$ and $P^*$), and the disk address $\texttt{a}_{v_{\text{last}}}$ of the last occupied leaf $v_{\text{last}}$ in the heap.
Finally, we write to $\texttt{addr}(i)$ the full precision data element $x_i$ and the disk address $\texttt{a}_v$ of $x_i$'s corresponding node $v$ in the heap.
We then discard all information on main memory except for the maximum element $x_{i^\star}$ and the disk addresses $\{\texttt{addr}(i)\}_{x_i \in X}$.

\textbf{Main Memory Space:} We use $\bits$ bits to store the maximum value $x_{i^*}$ and $\log N$ bits to store its ID $i^*$. (We also us $n \log N$ bits for the disk addresses $\{\texttt{addr}(i)\}_{x_i \in X}$ but we do not count these towards our final memory usage).

\textbf{Disk Space:} First, we verify that we have not assigned more than $O(\log N)$ words of information to any disk address.
For every $\addr(i)$ we need $O(1)$ words to store the full precision $x_i$ and the disk address $\heapaddr_v$ of $x_i$'s corresponding node $v$ in the heap.
For every $\heapaddr_v$ corresponding to a non-root node $v$, we need $O(\log n)$ words to store the key $(x_i, \addr(i))$ and the heap addresses $(\heapaddr_u)_{u \in P_v}$.
For the root node $v_r$, we also nee need $O(\log n)$ words to store $(x_{i^*}, \addr(i^*))$ and the heap addresses $(\heapaddr_u)_{u \in P_v}$ and $(\heapaddr_u)_{u \in P^*}$, and the disk address $\heapaddr_{v_{\text{last}}}$ of the last occupied leaf $v_{\text{last}}$.  
So all together we use $O(n \log N)$ words of disk space.

\textbf{Insertion:} Let $x_{n+1}$ be the point that is to be inserted into the heap.
It is not hard to see that the only information that is relevant to performing the standard max heap insertion procedure are the keys on the insertion path $P^*$.
It is also not hard to see that after an insertion, the only points whose deletion path change are those on $P^*$ and the only information required to update these deletion paths is the deletion paths of all nodes on $T(P^*)$.
Thus, our goal will be to use as few I/Os as possible to recover this information.

To this end, we perform a sequence of 3 read I/Os (of different sizes) as follows. We first read the disk address $\addr(i^*)$ of the maximum element $x_{i^*}$.
From this I/O we get the disk address $\heapaddr_{v_r}$ of the root node $v_r$ of the max heap.
We then perform a second I/O to obtain the contents of this address from which we get the disk addresses $(\heapaddr_u)_{u \in P^*}$ corresponding to the nodes on the insertion path.
Finally, we can probe these heap node disk addresses as well as the disk addresses of all nodes that are siblings to $P^\star$\footnote{We will assume that the disk addresses $\{\heapaddr_u\}$ are sufficiently structured so as to easily infer the address of the sibling node to any node $v$ for which we have its corresponding disk address $\heapaddr_v$ (if such a sibling exists).} to get two things: 1) the keys $\{(x_j , \addr(j)) : \text{$x_j$'s node is in } T(P^*)\}$ and 2) the sequence of disk addresses $(a_u)_{u \in P_v}$ corresponding to the deletion path of \emph{every} $v$ in $T(P^*)$.
In other words, from the third probe we have recovered the contents of the off-by-one tree $T(P^\star)$ in main memory as well as the (disk addresses) of the deletion paths of all nodes in $T(P^*)$.

Next, we assign disk addresses $\addr(n+1)$ and $\heapaddr_{v_{\text{new}}}$ assigned to $x_{n+1}$ and the first vacant leaf $v_{\text{new}}$ that we fill with the key $(x_{n+1}, \addr(n+1))$ respectively.
We can then perform the standard insertion procedure for $x_{n+1}$ into the heap since we have recovered all the relevant local information. 
We can then update $x_{i^*}$ in main memory if it has changed.
This sequence of steps takes no more than $O(\log n \log N)$ time.
Notice that this operation preserves dynamic consistency since the maximum element is unique since we break ties using the numerical values of the IDs.

After the insertion is completed, we need to update the disk with these changes. 
First, we need to update the contents of any disk address $\heapaddr_v$ corresponding to non-root nodes $v$ in $P^\star$ whose contents have changed (because we have moved around keys in the heap and/or have new deletion paths).
This requires $O(\log n)$ write probes.
Second, we need to update the contents of the address $\heapaddr_{v_r}$ corresponding to the root node $v_r$ to update the potentially new maximum encoded in the key tuple $(x_{n+1}, \addr(n+1))$, the new insertion path,\footnote{Again, we will assume that the disk addresses for the heap are sufficiently structured so that we can infer this information from the addresses of the nodes on the insertion path $P^\star$.} and the disk address $\heapaddr_{v_{\text{last}}}$ corresponding to the new last occupied leaf $v_{\text{new}}$.
This requires one write probe of size $O(1)$.
Finally, we need to write to $\addr(n+1)$ the full precision of the point $x_{n+1}$ and the disk address $\heapaddr_{v*}$ correspond to the node $v^*$ in $P^*$ that it ends up in after insertion which requires a write probe of size $O(1)$.
Thus, we can request a single batch of $O(\log n)$ write probes to apply all the changes and it takes $O(\log n \log N)$ time to prepare all these write probes.
All together, we conclude that the running time of the insertion procedure is $O(\log n \log N)$.

\textbf{Deletion:}
Let $x_i$ be the point that is to be deleted from the dataset.
The high level argument is slightly more complicated than the insertion procedure. 
We perform a sequence of 3 read I/Os (of different sizes) as follows. 
We first read the disk address $\addr(i)$ of the to-be deleted point $x_{i}$.
From this I/O we get the disk address $\heapaddr_{v}$ of the root node $v$ corresponding to $x_i$ in the max heap.
We then perform a second I/O to obtain the contents of this address from which we get the disk addresses $(\heapaddr_u)_{u \in P_v}$ corresponding to the nodes on the deletion path $P_v$.
Finally, we can probe these heap node disk addresses, the disk addresses of all nodes that are siblings to $P_v$ and the address of the last occupied leaf $\heapaddr_{v_{\text{last}}}$ to get three things: 1) the keys $\{(x_j , \addr(j)) : \text{$x_j$'s node is in } T(P_v)\}$, 2) the sequence of disk addresses $(a_u)_{u \in P_v}$ corresponding to the deletion path of \emph{every} $v$ in $T(P^*)$, and 3) the key $(x_k, \addr(k))$ of the last occupied leaf $v_{\text{last}}$.
In other words, from the third probe we have recovered the contents of the off-by-one tree $T(P^\star)$ in main memory, the (disk addresses) of the deletion paths of all nodes in $T(P^*)$ and the point $(x_k, \addr(k))$ that replaces $(x_i, \addr(i))$ in the heap at the beginning of the deletion operation.

We can then perform the standard deletion procedure to remove $x_{i}$ from the heap since we have recovered all the relevant local information. 
We can then update $x_{i^*}$ in main memory if it has changed.
This sequence of steps takes no more than $O(\log n \log N)$ time.
Notice that this operation preserves dynamic consistency since the maximum element is unique since we break ties using the numerical values of the IDs.

After the deletion is completed, we need to update the disk with these changes. 
First, we need to update the contents of any disk address $\heapaddr_v$ corresponding to non-root nodes $v$ in $P_v$ whose contents have changed (because we have moved around keys in the heap and/or have new deletion paths).
This requires $O(\log n)$ write probes.
Second, we need to update the contents of the address $\heapaddr_{v_r}$ corresponding to the root node $v_r$ to update the potentially new maximum, the new insertion path (the root-to-leaf path of to the parent of the former last occupied leaf $v_{\text{last}}$), and the disk address $\heapaddr_{v_{\text{last}}}$ corresponding to the new last occupied leaf.
This requires one write probe of size $O(1)$.
Third, we need to update the disk address $\addr(k)$ corresponding to the point $x_k$ that was formerly inside the last occupied leaf $v_{\text{last}}$ which requires a write probe of size $O(1)$.
Finally, we need to delete/free the contents of $\addr(i)$ which also requires a write probe of size $O(1)$.
Thus, we can request a single batch of $O(\log n)$ write probes to apply all the changes and it takes $O(\log n \log N)$ time to prepare all these write probes.
All together, we conclude that the running time of the insertion procedure is $O(\log n \log N)$.

\textbf{Query:} This takes $O(1)$ time as we store the maximum element $x_{\text{max}}$ explicitly in main memory.    
\end{proof}

\subsection{\ouralg}\label{app:tbq_tbq}

\begin{algorithm}
\caption{\textsc{\ouralg~ Build}}\label{alg:tbq_build}
\textbf{Inputs:} Dataset $X = \{x_1, \dots, x_n\} \in 
 \mathbb{R}^{d \times n}$ and encoding length $L$. \textbf{Output:} Encoded dataset $X_{\text{Enc}} \in [2^L]^{n}$,  codebook $\mu \in \mathbb{R}^{2^L \times d}$, rotation matrix $\mc{O} \in \mathbb{R}^d$, DySk priority queues $\{H_i\}_{i \in [2^{L+1}-2]}$.
\vspace{0.5em}

\begin{algorithmic}
    \STATE Initialize dataset encoding $X_{\text{Enc}}$ and product codebook $\{ \mu_j \}_{j=1}^{M}$.\vspace{0.5em}
    \STATE Sample a random rotation matrix $\mathcal{O} \in \mathbb{R}^{d \times d}$.\vspace{0.5em}
    \STATE Let  $X' = \{\mathcal{O} x_1, \dots, \mathcal{O} x_n\}$ be the randomly rotated dataset.
    \vspace{0.5em}
    \STATE Sample a random sequence $S = (i_1, \dots, i_L)$ from $\{1, \dots, d\}$ without replacement.\vspace{0.5em}
    \STATE Compute kd-tree $T $ using $X'$ and $S$.\vspace{0.5em}
    \STATE Determine the partition $\{\mc{P}_{k}\}_{k \in [2^L]} \subset [N]$ via the leaves of  $T$.
    \vspace{0.5em}
        \FOR{$k$ in $[2^L]$}\vspace{0.5em}
        \STATE Set $k$th row of $\mu$ to be the mean of all points in $X$ that are assigned to $\mc{P}_{k}$.\vspace{0.5em}
        \ENDFOR
        \vspace{0.5em}
        \FOR{ $i$ in $[n]$}\vspace{0.5em}
        \STATE Set $i$th entry of $X_{\text{Enc}}$ to the identity $k$ of $x_i$'s group $\mc{P}_{k}$.
        \ENDFOR
    \vspace{0.5em}
    \FOR{non-root nodes $v_i$ in $T$}
        \IF{$v_i$ is the left child of its parent}
        \STATE Construct a DySk max priority queue $H_i$ on points in $v_i$.
        \ELSE
        \STATE Construct a DySk min priority queue $H_i$ on points in $v_i$.
        \ENDIF
        \ENDFOR
    \STATE Store auxiliary information for priority queues on disk (as per Lemma \ref{lem:max_DDS}).
    \vspace{0.5em}
    \STATE \textbf{return} $X_{\text{Enc}}$,  $\mu$, $\mc{O}$, and $\{H_i\}_{i \in [2^{L+1}-2]}$.
\end{algorithmic}
\end{algorithm}

\begin{algorithm}
\caption{\textsc{\ouralg~ Insert}}\label{alg:insert}\textbf{Inputs:} New point $x_{n+1} \in 
 \mathbb{R}^{d}$.\\
\textbf{Output:} Updated encoding $X_{\text{Enc}} \in [2^L]^{n+1}$, code book $ \mu \in \mathbb{R}^{2^L \times d}$, and priority queues $\{H_i\}_{i \in [2^{L+1}-2]}$.
\vspace{0.5em}

\begin{algorithmic}
\STATE Append new row to dataset encoding $X_{\text{Enc}}$.\vspace{0.5em}
    \STATE Reconstruct partition $\{\mc{P}_{k}\}_{k \in [2^L]}$ using encoding $X_{\text{Enc}}$.\vspace{0.5em}
    \STATE Construct new partition $\{\mc{P}'_{k}\}_{k \in [2^L]}$ using max/min elements stored in priority queues and $x_{n+1}$.\vspace{0.5em}
        \FOR{non-root node $v_i$ in $T$}\vspace{0.5em}
            \IF{$v_i$ is the left child of its parent}\vspace{0.5em}
            \STATE $y_i\leftarrow$ maximum value from max priority queue $H_i$.\vspace{0.5em}
            \ELSE
            \STATE $y_i\leftarrow$ minimum value from min priority queue $H_i$.\vspace{0.5em}
            \ENDIF\vspace{0.5em}
            \IF{$y_i$'s encoding has changed}\vspace{0.5em}
            \STATE Subtract $y_{i}$'s contribution from the rows of $\mu$ determined by the old partition $\{\mc{P}_{k}\}$.\vspace{0.5em}
            \STATE Add $y_i$'s contribution to the rows of $\mu$ determined by the new partition $\{\mc{P}'_{k}\}$.\vspace{0.5em}
            \ENDIF
        \ENDFOR
    \STATE Add $x_{n+1}$'s contribution to the row of $\mu$ determined by the new partition $\{\mc{P}'_{k}\}$.\vspace{0.5em}
    \STATE set $(n+1)$th row of $X_{\text{Enc}}$ to encoding of $x$'s group $\mc{P}'_{k}$.\vspace{0.5em}
    \STATE Udpate the priority queues $\{H_i\}_{i \in [2^{L+1}-2]}$ using the DySk-based update defined in \cref{lem:max_DDS}.\vspace{0.5em}
    \STATE \textbf{return} $X_{\text{Enc}}$, $\mu$ and $\{H_i\}_{i \in [2^{L+1}-2]}$.
\end{algorithmic}
\end{algorithm}

\begin{algorithm}
\caption{\textsc{\ouralg~ Delete}}\label{alg:delete}
\textbf{Inputs:} Point to be deleted $x_{i} \in 
 \mathbb{R}^{d}$.\\
\textbf{Output:} Updated encoding $X_{\text{Enc}} \in [2^L]^{n-1}$, code book $ \mu \in \mathbb{R}^{2^L \times d}$, and priority queues $\{H_i\}_{i \in [2^{L+1}-2]}$.
\vspace{0.5em}

\begin{algorithmic}
\STATE Delete $i$th row from dataset encoding $X_{\text{Enc}}$.\vspace{0.5em}
    \STATE Reconstruct partition $\{\mc{P}_{k}\}_{k \in [2^L]}$ using encoding $X_{\text{Enc}}$.\vspace{0.5em}
    \STATE Construct new partition $\{\mc{P}'_{k}\}_{k \in [2^L]}$ using max/min elements stored in priority queues and $x_{i}$.\vspace{0.5em}
        \FOR{non-root node $v_i$ in $T$}\vspace{0.5em}
            \IF{$v_i$ is the left child of its parent}\vspace{0.5em}
            \STATE $y_i\leftarrow$ maximum value from max priority queue $H_i$.\vspace{0.5em}
            \ELSE
            \STATE $y_i\leftarrow$ minimum value from min priority queue $H_i$.\vspace{0.5em}
            \ENDIF\vspace{0.5em}
            \IF{$y_i$'s encoding has changed}\vspace{0.5em}
            \STATE Subtract $y_{i}$'s contribution from the rows of $\mu$ determined by the old partition $\{\mc{P}_{k}\}$.\vspace{0.5em}
            \STATE Add $y_i$'s contribution to the rows of $\mu$ determined by the new partition $\{\mc{P}'_{k}\}$.\vspace{0.5em}
            \ENDIF
        \ENDFOR
    \STATE Subtract $x_{i}$'s contribution to the row of $\mu$ determined by the old partition $\{\mc{P}_{k}\}$.\vspace{0.5em}
    \STATE Update the priority queues $\{H_i\}_{i \in [2^{L+1}-2]}$ using the DySk-based update defined in \cref{lem:max_DDS}.\vspace{0.5em}
    \STATE \textbf{return} $X_{\text{Enc}}$, $\mu$ and $\{H_i\}_{i \in [2^{L+1}-2]}$.
\end{algorithmic}
\end{algorithm}

In this subsection we include the proof of \cref{thm:tbq}, our main result for \ouralg.
\begin{proof}[Proof of \cref{thm:tbq}]
Define $I = I(n) = \lceil n/2 \rceil$ to be the right median index and $I-1$ to be the left median index of a dataset of size $n$.
The fundamental operations for a \ouralg~are explained below.

\textbf{Build:} 
Given an initial dataset $X = \{x_1, \dots, x_n\}$ our goal will be to compute an $L$ bit encoding for each $x_i$ and construct a corresponding code books $\mu \in \mathbb{R}^{2^L \times d}$ whose entries are in one-to-one correspondence with the set of possible encodings.
We begin by sampling a random rotation matrix $\mc{O}$ and rotate all points in $X$ to obtain the augmented dataset $X'$.
We then construct a depth $L+1$ kd-tree on $X'$. 
Let $S = (j_1, \dots, j_L) \in [d]^\ell$ be the sequence of dimensions that the kd-tree uses to split on.
We assign $L$ bits to every $x_i$ based on the root-to-leaf path of its rotation $\mc{O}x_i$ in the kd-tree. 
We set the $j$th row of the codebook $\mu$ to be the average of all (un-rotated) points in $X$ that land in the $j$th leaf after being rotated.

We now explain how to incorporate the priority queues from \cref{lem:max_DDS} to enable efficient dynamic updates.
Recall that every non-root node $v$ in the kd-tree corresponds to a subset $X_v'$ of the dataset $X'$.
If $v$ is a left child at depth $\ell$, we build a max priority queue from \cref{lem:max_DDS} on $X_v'$ where the order of the points $X_v'$ is determined by the index $j_{\ell-1}$.
If $v$ is a right child, we build a minimum priority queue on $X_v'$ instead.
We will make the harmless assumption that every disk address points to a block of memory that consists of $\Theta(\max\{d, L,\log N\})$ words.
Notice that a naive implementation of the priority queues would require us to store $L$ disk addresses $\addr(i)$ for each $x_i$ since $x_i$ participates in $L$ priority queues.
To avoid this, we apply \cref{lem:max_DDS} in a ``grey box'' way as follows: for each $x_i$ we store in $\addr(i)$ the \emph{both} the full precision of the original vector $x_i$ and its rotated version $\mc{O}x_i$ as well as \emph{all} disk addresses $\{\texttt{a}_v\}$ of $x_i$'s corresponding nodes $\{v\}$ in all $L$ heaps used in the priority queues.

We then discard all information stored on main memory except for the rotation matrix $\mc{O}$, the max/min elements (and their identities) of all  priority queues, the number of vectors in every priority queue and the disk address $\{\addr(i)\}_{x_i \in X}$. 
The pseudocode for this operation is provided in \cref{alg:tbq_build}.

\textbf{Main Memory Space:} 
We use $d^2 b$ bits to store the rotation matrix $\mc{O}$.
We use $2^L d \bits $ bits to store the code book $\mu$.
We use $n L$ bits to store the encoding of all vectors $x_i$s.
We also use $(2^{L+1}-2)(\log N + db)$ bits to store all the max/min elements (and their identities) for all priority queues.
These max/min are the rotated points $\{\mc{O}x_i\}$.
We also use no more than $(2^{L+1}-2)\log N$ bits to store the number of vectors in the priority queues.
So in total we use less than $d^2b + 2^{L+1}(2\log N + d\bits) + 2^L d \bits + nL$ bits of space in main memory.
(We also use $ n \log N$ bits of space for all the disk addresses corresponding to the $x_i$s but we do not count this).

\textbf{Disk Space:}
Let us verify that we have not assigned more than $O(\max\{d, L,\log N\})$ words of information to any disk address.
Recall that for each $x_i$ we have modified the contents of the disk address $\texttt{addr}(i)$ to include \emph{all} disk addresses $\{\texttt{a}_v\}$ of $x_i$'s corresponding nodes $\{v\}$ in all $L$ heaps used in the priority queues. 
Thus we only need $O(L)$ words to store these points.
The full precision vectors are now $d$ dimensional vectors so we need $O(d)$ words of memory to store them.
In the proof of \cref{lem:max_DDS}, we also store disk addresses $\heapaddr_v$ (corresponding to some vertex $v$ in some heap) which also require $O(d)$ words to store the full precision value of the corresponding key $(x_i, \addr(i))$.
Taking into account our modification to \cref{lem:max_DDS} explained above, we can conclude that we do not use more than $O(n 2^L\max\{d, L, \log N\})$ words of disk space.

\textbf{Insertion:}
Let $x_{n+1}$ be the point that is to be inserted into $X$ and let $x'_{n+1} = \mc{O}x_{n+1}$ be the vector after applying the random rotation stored in main memory.
Recall that we have stored all the max/min values of all priority queues in main memory as the number of points in each priority queue.
Thus, using all this stored information and $x'_{n+1}$, we can determine which points have changed their root-to-leaf paths using the guarantees of \cref{prop:stability_median,thm:kd-tree-stability}.
For every point that has changed, we can add and subtract its contribution from its new and old leaf respectively.
We can then update the encoding of every point that has changed based on which leaf it now lives in.
We can also determine which leaf $x_{n+1}$ lands in and add its contribution to it. 
This sequence of operations takes $O(d^2 + 2^L d)$ time.
Notice that this operation preserves dynamic consistency since the maximum/minimum elements of every priority queue are unique since we break ties using the numerical values of the IDs.

After updating the encodings and code books, we need to update the priority queues.
From the guarantee of \cref{thm:kd-tree-stability} we know that every priority queue gets at most one insert and/or one delete.
Moreover, the points that are to be inserted/deleted are already in main memory. 
We can thus apply the appropriate update (insert or delete) from \cref{lem:max_DDS} to all the priority queues in parallel.
We can slightly modify the delete operation of \cref{lem:max_DDS} in the case of a simultaneous insert and delete.
The final I/O size claimed is inherited from the guarantee of \cref{lem:max_DDS} and our ``grey box'' modification described in the build operation. 
Putting everything together we have that the running time of the insertion procedure is $O(d^2 + 2^L \log n \max\{d, L, \log N\})$.
The pseudocode for this operation is provided in \cref{alg:insert}.

\textbf{Deletion:}  
The deletion operation is very similar to the insertion operation.
Let $x_i \in X$ be the point that is to be deleted.
We can determine which points in the tree change their root-to-leaf path and by the same argument as the insertion operation and we can update the code book and all the encodings the same way.
The main difference is that we remove $x_i$'s contribution from its leaf but don't add it to any other leaf since it is being deleted.
Updating the priority queues is also identical to the insertion operation.
The pseudocode for this operation is provided in \cref{alg:delete}.

\textbf{Query:} Given a query $q \in \mathbb{R}^d$, we can return our estimate for the $k$-nearest neighbors as follows.
We first compute the distances between every $\mu_j \in \mathbb{R}^d$ (from the code book $\mu$) and the query point $q$ which takes $O(2^Ld)$ time.
We can then keep track of the $k$ nearest neighbors via brute force search (breaking ties in an arbitrary deterministic way) which takes an additional $O(d2^L + n\log k)$ time by computing a look-up table for the distances between the code words and the query point and then keeping track of the smallest $k$ distances using a simple minimum priority queue.
\end{proof}

\subsection{Details on Random Rotation}\label{app:rr}
As a preprocessing step, \ouralg\ applies a random rotation the given search dataset. The rotation is subsequently applied to every newly inserted data point and to every query point. 
This is a common preprocessing step, often used in ANN and in many other computational geometry problems. 

A random rotation is a change of basis on the input vectors, from the basis in which they were originally given to a uniformly random basis. In other words, the vectors are re-written in a new, uniformly random system of coordinates. On the one hand, a change of basis is an isometry; it perfectly preserves the distances between all pairs of vectors, so nearest neighbor searches are not affected. On the other hand, it ``smooths'' the representation of the vectors across coordinates, and assists in avoiding possible pathologies in the input system of coordinates. After a random change of basis, each new coordinate is an identically distributed linear combination of the original coordinates.

\subsection{\ouralg~ based Product Quantizer}\label{app:tbq_tbpq}
In this subsection we introduce some relevant background on Product Quantization (PQ) \cite{vqindex2002, jegou2010product}.
We then show how PQ can be modified by incorporating our \ouralg~  data structure in place of k-means cluster to obtain \ouralg-PQ.
\subsubsection{A Primer on Product Quantization}\label{app:tbq_pq}
Product quantization (PQ) is practical and widely deployed method for vector quantization.
Given a dataset $X = \{x_1, \dots, x_n\}$ 
of $d$ dimensional vectors, a product quantizer with parameters $M$ and $L$ splits every vector $x = (x^1, \dots, x^M) \in X$ into $M$ contiguous sub-vectors, each of dimension $d/M$, and forms the block datasets $X_1, \dots, X_M$ from these sub-vectors.
For each block dataset $X_j$ it computes a $k$-means clustering of $X_j$ to obtain $k=2^L$ centroids $\mu_j = \{\mu_{j,1}, \dots, \mu_{j,2^L}\}$.
It then replaces every sub-vector $x_i^j \in X_j$ with the identity $t \in [2^L]$ of its nearest centroid $\mu_{j,t}$.
Thus, every vector $x_i = (x_i^1, \dots, x_i^M) \in X$ is replaced with the compressed representation $C_i = (t_1, \dots, t_M) \in [2^L]^M$.
Computing the approximate distance between a search point $q \in \mathbb{R}^d$ and $x_i \in X$ we can ``decompress" $C_i$ into $\hat{x}_i = (\mu_{1, t_1}, \dots, \mu_{M, t_M}) \in \mathbb{R}^d$ and compute the distance $\|\hat{x}_i - q\|_2$.
Thus all we need to store in main memory to answer query requests are the encodings $\{C_i\}_{i \in [n]}$ and the \emph{product codebook} $\{\mu_j\}_{j \in [M]}$.
See \cref{alg:PQ} for the pseudocode.
We can alternatively view the k-means clustering step as a way to partition the block dataset $X_j$ into $2^L$ groups $\{\mc{P}_{j,k}\}_{k \in [2^L]}$ for which we can represent every point $x_i^j$ in the group $\mc{P}_{j,k}$ by the average of all the points in $\mc{P}_{j,k}$.

\;
\begin{algorithm}
\caption{\textsc{Product Quantizer}}\label{alg:PQ}

\textbf{Inputs:} Dataset $X = \{x_1, \dots, x_n\} \in 
 \mathbb{R}^{d \times n}$,  number of blocks $M$, encoding length per-block $L$.\\
\textbf{Output:} Encoded dataset $X_{\text{Enc}} \in [2^L]^{n \times M}$ and product codebook $\{ \mu_j \}_{j=1}^{M} \subset \mathbb{R}^{2^L \times d/M}$.
\vspace{0.5em}

\begin{algorithmic}
    \STATE Initialize dataset encoding $X_{\text{Enc}}$.\vspace{0.5em}
    \STATE Initialize product codebook $\{ \mu_j \}_{j=1}^{M}$.\vspace{0.5em}
    \FOR{$j$ in $[M]$}
    \STATE Let  $X_j = \{x_1^j, \dots, x_n^j\}$ be the block dataset.
    \vspace{0.5em}
    \STATE Compute centroids $\{C_{j,k}\}_{k \in [2^L]}$ via k-means on $X_j$.\vspace{0.5em}
        \FOR{$k$ in $[2^L]$}
        \STATE Set $k$th row of $\mu_j$ to $C_{j,k}$.
        \ENDFOR
        \FOR{ $i$ in $[n]$}
            \STATE Set $(i,j)$th entry of $X_{\text{Enc}}$ to $k^\star = \text{argmin}_{k \in [2^L]} \|C_{j,k} - x_i^j\|_2$.
            \ENDFOR
    \ENDFOR
    \vspace{0.5em}
    \STATE \textbf{return} $X_{\text{Enc}}$ and $\{\mu_j\}_{j=1}^M$.
\end{algorithmic}
\end{algorithm}

\subsubsection{Product \ouralg~Guarantees}
We can replace the k-means step in PQ with \ouralg.
This enables us to build a PQ-like dynamically consistent quantization schemes that enjoy efficient updates.
The following theorem states the guarantees of this data structure. See~\cref{pq-tbq} for a high level illustration of this quantizer.

\begin{theorem}\label{thm:tbpq}
There is a dynamically consistent implementation of a product \ouralg~with parameters $L$ and $M$ with the following guarantees:
    \begin{CompactItemize}
        \item \textbf{Main memory space:} $nL +  2^L d \bits + 2^{L+1}M(2\log N + d\bits)+ Md^2$ bits of space.
        \item \textbf{I/Os per-update:}
        A sequence of three read I/Os of size $O(T)$, $ O(T)$, $O(T\log n)$ words and a single write I/O of size $O(T\log n)$ words where $T = M2^L\max\{d, ML, \log N\}$.
        \item \textbf{Update time:} $O\left(\frac{d^2}{M} + 2^L M\log n \max\{d, ML, \log N\}\right)$ time.
        \item \textbf{Query time:} 
        $O(d2^L + n(M + \log k))$ time to return $k$-approximate nearest neighbors.
    \end{CompactItemize}
\end{theorem}
Implementing \ouralg-PQ requires a simple modification to the implementation of \ouralg~appearing in \cref{thm:tbq}.
For every point $x_i$ we need to modify the contents of $\addr(i)$ to contain the heap disk addresses $\{\heapaddr_v\}$ to all $ML$ heap node that $x_i$ belongs to across all blocks.
The query algorithm is the standard one used for PQ. (See \cite{jegou2010product}).
Since the modification is straightforward and the proof is nearly identical to that of \cref{thm:tbq}, we omit the proof of \cref{thm:tbpq}.

\section{Experiment Details}\label{app:experiments}

\begin{algorithm}[t!]
\caption{\textsc{Construct Streaming Dataset}}\label{alg:cdss}

\textbf{Inputs:} Static dataset $X$, number of clusters $c$, initial dataset size $n_0$, query fraction $f_q \in (0,1)$, iterations per cluster $\tau$, query freshness $\alpha \in [0,1]$, delete fraction $f_d \in [0,1]$. \vspace{0.5em}\\
\textbf{Output:} Streaming dataset $\{(U_t, Q_t)\}_{t=0}^T$.
\vspace{0.5em}

\begin{algorithmic}
    \STATE Partition $X$ into $c$ clusters $C_1, C_2, ..., C_c$ via k-means. Shuffle cluster order and find $j = \min\{k : \sum_{i\leq k} \vert C_i \vert \geq n_0\}$.\vspace{0.5em}
    \STATE Sample an $f_q$-fraction of elements uniformly at random from $\cup_{i\leq j} C_i$ as $Q_0$, and define the remainder as $U_0= I_0$. \vspace{0.5em}
    \STATE Re-order the remaining clusters $C_1, ..., C_{c-j}$ by increasing order of the distance between their centroid and that of $I_0$.\vspace{0.5em}
    \STATE Sample from each remaining cluster an $f_q$-fraction of elements uniformly at random as the candidate query set $\tilde{Q}_i, i=1, ..., c-j$. Define $\tilde{Q}_0 = Q_0$.\vspace{0.5em}
    \STATE Partition the remainder of each cluster into $\tau$ (approximately) even subsets $C_{ik}, k = 1, ..., \tau$. \vspace{0.5em}
    \STATE Set $t=1$. \vspace{0.5em}
    \FOR{$i = 1, ..., c-j$}\vspace{0.5em}
        \STATE Associate to each $\tilde{Q}_j$, $j\leq i$, the weight $w_i = (1-\alpha)^{\vert j - i \vert}$, applying the convention $0^0 = 1$ where appropriate. \vspace{0.5em}
        \FOR{$k = 1, ..., \tau$}
            \STATE Set $I_t = C_{ik}$. \vspace{0.5em}
            \STATE Set $D_t$ as the $\lceil{f_d \vert I_t \vert}\rceil$-oldest elements in $X_t$, breaking ties uniformly at random. \vspace{0.5em}
            \STATE Generate $Q_t$ as a sample of size $\lceil{f_q \vert I_t \vert}\rceil$ from $\cup_{j \leq i} \tilde{Q}_j$, sampling without replacement and weighting queries in each $\tilde{Q}_j$ by $w_j$. \vspace{0.5em}
            \STATE Set $U_t = (I_t, D_t)$. \vspace{0.5em}
            \STATE Set $t = t+1$. \vspace{0.5em}
        \ENDFOR
    \ENDFOR    
    \vspace{0.5em}
    \STATE \textbf{return} $\{(U_t, Q_t)\}_{t=0}^T$.
\end{algorithmic}
\end{algorithm}

\paragraph{Dataset construction for query freshness experiment.}

We construct the dynamic search scenario from our 100K sample of Deep1B, taking $n_0$ to be 10\% of the data, $C=10$, $\tau=10$, $f_q = 0.1$, and $f_d = 0$. We vary the query freshness parameter $\alpha$ across three values: $\alpha=0$ (no freshness, queries uniform over the observed clusters), $\alpha=0.1$ (queries slightly biased towards more recent clusters), and $\alpha=1$ (all queries sampled from the most recent cluster). We set $M=8$ and $L=12$ for the RebuildPQ and FrozenPQ quantizers, for a compression rate of $1$ bit per dimension.

\paragraph{Dataset construction for recall experiments.}

We construct streaming datasets from each of the three 100M-scale static vector datasets described in \cref{sec:experiments}. For each, we take $n_0$ to be 10\% of the data, $C=10$, $\tau=10$, $f_q = 0.1$, $f_d = 1$, and $\alpha=1$. For each dataset, all product quantizers use the same number of blocks $M$ and codebook size $L$. We use $M=8$, $L=12$ for Deep1B; $M=8$, $L=12$ for BigANN; and $M=10$, $L=12$ for Text2Image. This implies a compression to $1$, $0.75$, and $0.6$ bits per dimension, respectively.

DeDriftPQ requires the additional parameter $m$, which controls how many of the largest clusters are selected for repartitioning. Note that the \emph{total} number of clusters to be repartitioned is strictly greater, as an additional number $m'$ of the smallest clusters are included to keep the total size of the clustering fixed over time. We use $m=2$ in this section. Although it may seem low, we make two observations to the contrary: first, at codebook size $L=2$, $m=2$ is proportionally much larger than the settings explored by the authors in \cite{dedrift}, which favors the method in terms of quantization performance; second, even at low $m$ the number of vectors reassigned can be significant relative to the total size of the data - see \cref{fig:iocosts} in the main paper.

\paragraph{I/O cost and latency experiment details.} 

No streaming scenario is constructed for the I/O cost experiments, since we are simply computing the cost of a single update. Datasets are drawn from Deep1B, and as above we use $M=8$ for both the \ouralg~ and DeDriftPQ quantizers. We vary $L \in \{4, 6\}$. We set $m$ for DeDrift as a function of the $L$, specifically the smallest integer such that at least $2\%$ of clusters are reassigned at each iteration; this yields $m=1$ for $L=4$ and $m=2$ for $L=6$. Note that in this experiment, smaller $m$ favors DeDrift, and by definition $m\geq 1$.

\paragraph{Counting Disk I/O Costs}\label{app:cdioc}

In \cref{sec:experiments}, we compare disk I/O costs of \ouralg~versus DeDriftPQ for a single update over varying scales for the full dataset. As discussed in that section, we use ``number of full-precision vectors required from disk to make an update" as our practical instantiation of the disk I/O cost. Here we describe in detail how we compute this number for both methods. 

\textbf{DeDriftPQ} We begin by learning a PQ codebook on the dataset. We subsequently quantize the entire dataset using the learned codebook. Then, given an update vector and parameter $m$:
\begin{enumerate}
    \item Encode the update vector using the existing codebook and update cluster counts accordingly.
    \item Apply the DeDrift criterion independently in each of the $M$ blocks to identify a subset of data for reclustering.
    \item Take the union of all unique vector IDs returned across blocks, and count the size of this set.
\end{enumerate}

\textbf{\ouralg} We use the read I/O bound computed in  \cref{thm:tbpq} as a starting point to compute the cost of a streaming update, but replace the worst-case term for the number of heap insertion / deletion paths required from disk ($2^L$) with an empirical estimate from our sampled datasets. This proceeds as follows:
\begin{enumerate}
    \item Initialize \ouralg~ on the dataset.
    \item Within each block, apply the \ouralg~ update for the new (sub)vector to each KD-tree data structure:
    \begin{enumerate}
        \item Starting at the root node for each tree, check whether the min heap, max heap, or both require an update upon insertion. This depends on the value of the inserted point and the location of the current median; see \cref{app:tbq_tbq}.
        \item Add the count (0, 1, or 2) of heaps changed, then recurse on the subtree(s) corresponding to the changed heaps.
        \item Return the total count of heaps changed for the tree.
    \end{enumerate}
    \item Return the sum of all heaps changed across blocks.
\end{enumerate}

\subsection{Further Experiments}\label{app:further}

\begin{wrapfigure}
{R}{0.45\textwidth}
\begin{center}
\includegraphics[width=0.45\textwidth]{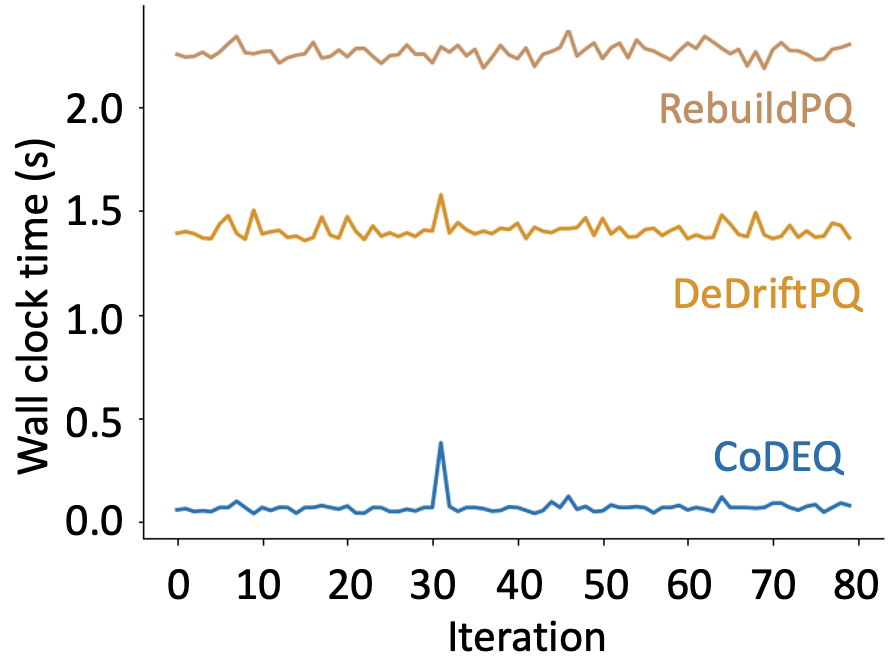}
\caption{Wall-clock latency of single point insertion iterations.}
\label{fig:latency}
\end{center}
\end{wrapfigure}

\paragraph{Latency across streaming iterations.}

We complement the I/O cost evaluation with an empirical study of latency across streaming iterations in a simple implementation that communicates with disk. Our implementation is not optimized for performance (which typically involves many tricks to limit or streamline communication with disk) but rather designed to force disk reads and writes to illustrate differences in disk communication efficiency.

\paragraph{Approximate search indices.}

At massive scale, quantization is typically combined with an approximate search data structure such as hierarchical navigable small world (HNSW) graphs \cite{malkov2018efficient} or inverted file (IVF) indices \cite{jegou2010product}. We replace the flat index (i.e. exhaustive search) over the data with either HNSW or IVF to show the effect on the end-to-end recall performance in such systems. We use the FAISS implementations of both indices, setting $M=8$ in HNSW and using $128$ cells with $1$ probe for IVF. Results are reported in the top and middle rows of Figure~\ref{fig:recall_approx}, on the first 1000 iterations of the 100M datasets (roughly 3M points). We find that using \ouralg\ within the search pipeline leads to a durable advantage in terms of recall quality and stability over time.

\paragraph{Oversampling.}

In large-scale vector search with disk access, it is common to \emph{oversample}, that is, to retrieve a larger number of full-precision points from disk than required, then re-rank them in main memory. Oversampling can help compensate for the loss in recall incurred by nearest-neighbor search over the quantized points. We evaluate the impact of oversampling by computing the recall-10@50, representing the ideal performance of a two-step search procedure in which full-precision vectors are retrieved at a 5x oversampling rate. Results are reported in the bottom row of Figure~Figure~\ref{fig:recall_approx}. Across all datasets, \ouralg\ continues to match or outperform all baselines in this scenario.

\begin{figure*}[t]
\begin{center}
\includegraphics[width=\linewidth]{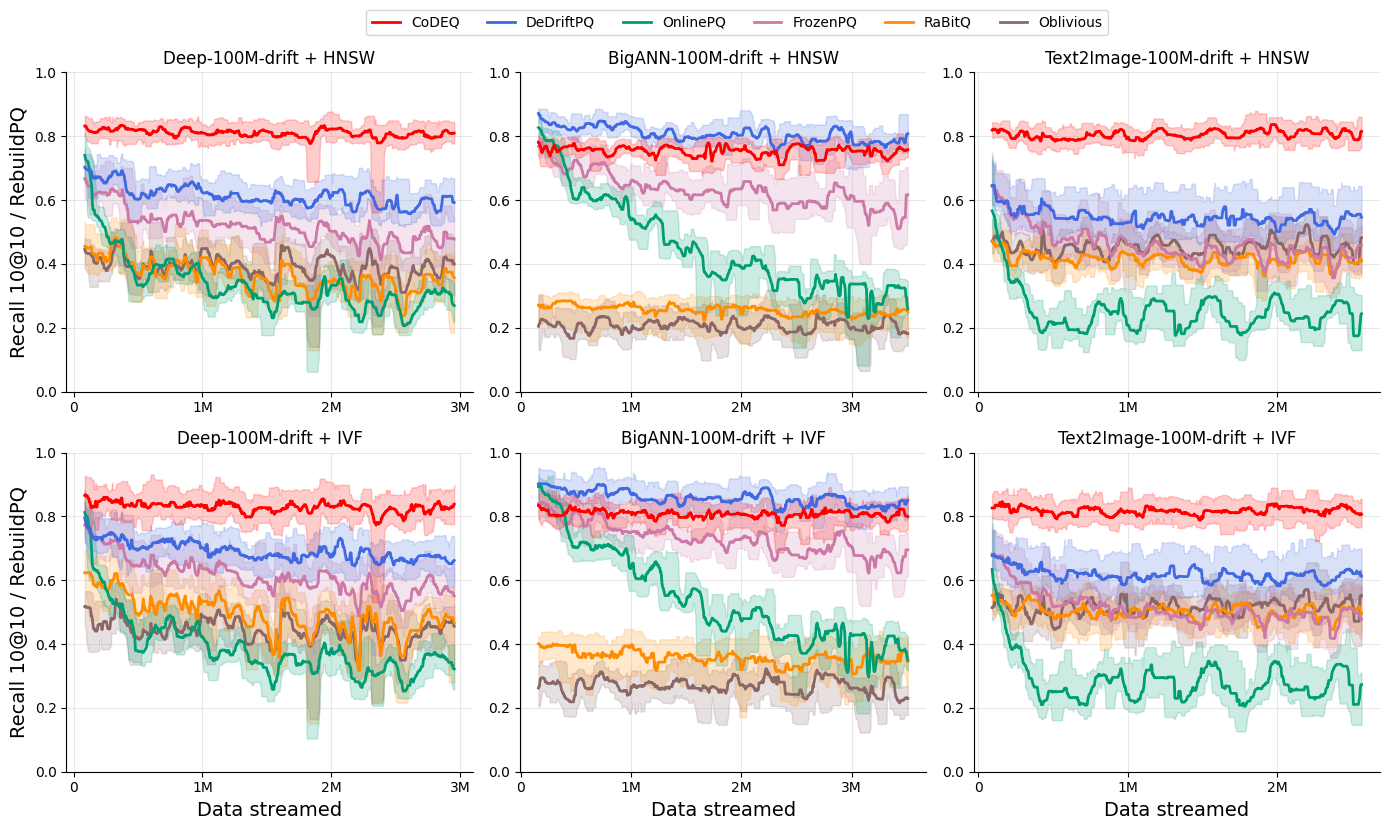}\\
\includegraphics[width=\linewidth]{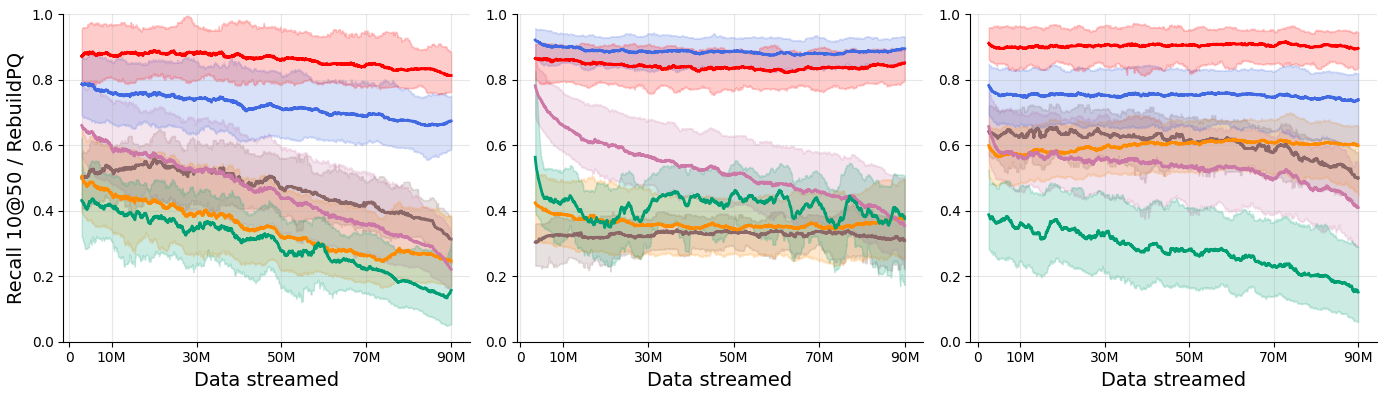}
\caption{Recall vs data streamed for quantizers combined with additional practical components of a search pipeline.
Solid lines denote the rolling median while shaded regions capture the rolling 10th-90th quantiles. Top row: Approximate k-NN search with HNSW. Middle row: Approximate k-NN search with IVF. Bottom row: Recall-10@50 expressed as a fraction of RebuildPQ, representing a 5x oversampling approach.}
\label{fig:recall_approx}
\end{center}
\end{figure*}

\begin{wrapfigure}
{R}{0.45\textwidth}
\begin{center}
\includegraphics[width=0.45\textwidth]
{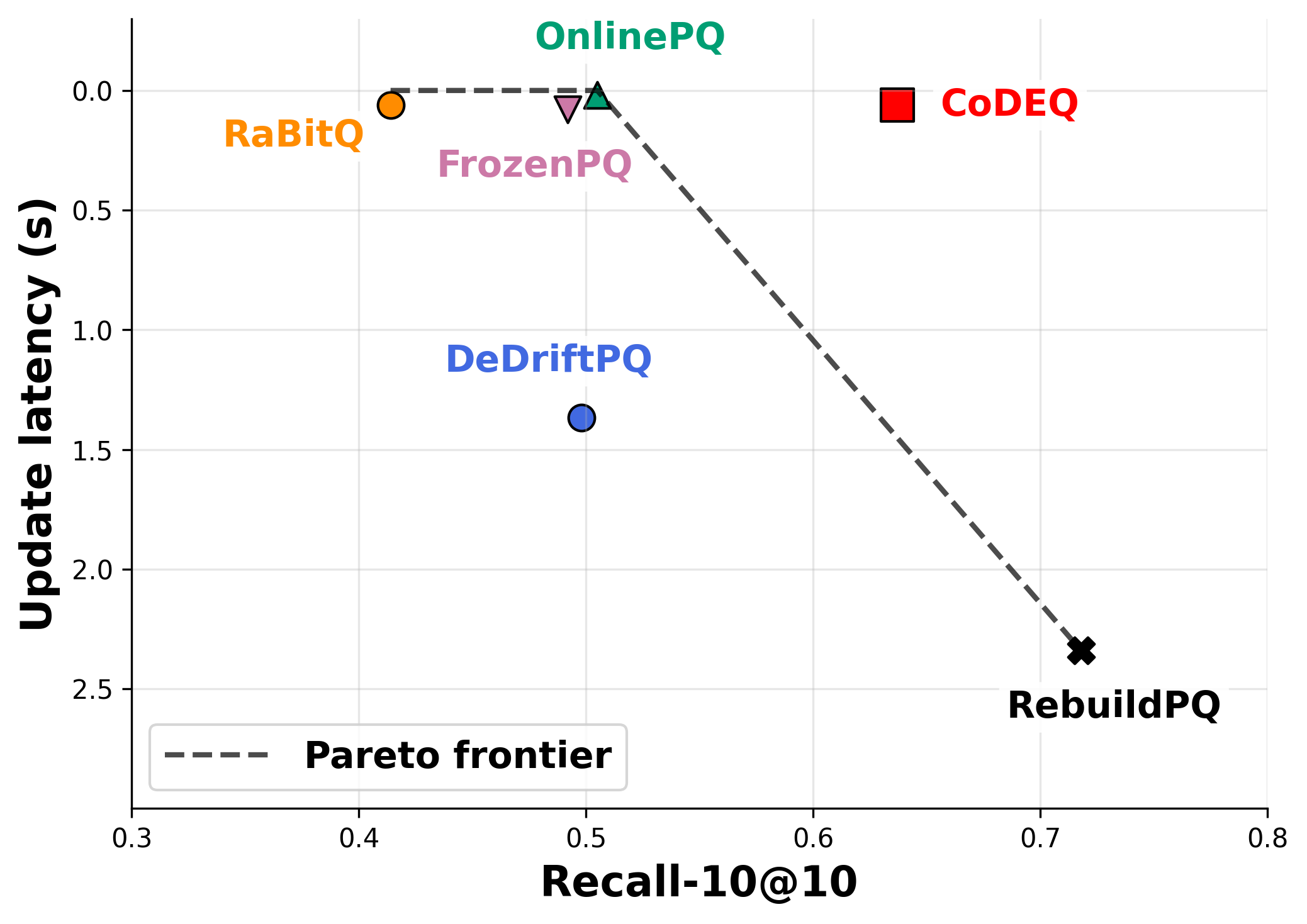}
\caption{CoDEQ extends the recall-latency Pareto frontier for streaming quantization on drifting data. Points reflect performance of each quantizer on the Deep-\emph{drift} dataset, whose construction is described in Section~\ref{sec:experiments}.}
\label{fig:pareto}
\end{center}
\end{wrapfigure}

\title{Quantization for Vector Search under Streaming Updates}

\end{document}